%% file: arxiv.tex
\documentclass[journal,twoside,web]{ieeecolor}
\usepackage{generic}
\usepackage{cite}
\usepackage{amsmath,amssymb,amsfonts}
\usepackage{graphicx}
\usepackage{textcomp}
\usepackage[T1]{fontenc}
\usepackage{hyperref}

\input{MACPdef}
\usepackage{bbold}

\usepackage{amsmath,amssymb,amsthm, amsfonts, mathrsfs}
\usepackage{graphicx}
\usepackage{xcolor}
\usepackage{comment}
\usepackage{dsfont}

\usepackage{cleveref}
\Crefname{equation}{}{}
\crefname{equation}{}{}
\usepackage{footmisc}
\usepackage{csquotes}
\let\labelindent\relax
\usepackage{enumitem}
\usepackage{bbm}

\usepackage{algorithm} 
\usepackage{algpseudocode}

\usepackage{soul} % for underline

\usepackage{caption, subcaption}
\def\BibTeX{{\rm B\kern-.05em{\sc i\kern-.025em b}\kern-.08em
    T\kern-.1667em\lower.7ex\hbox{E}\kern-.125emX}}

\newcommand{\ie}{\textit{i.e., }}
\newcommand{\eg}{\textit{e.g., }}
    
\newcommand{\R}{\mathbb R}    
   
\newcommand\eqnumber{\addtocounter{equation}{1}\tag{\theequation}}

\DeclareMathOperator*{\argmax}{arg\,max}
\DeclareMathOperator*{\argmin}{arg\,min}
\newcommand{\minf}{\Phi}% min phi 
\newcommand{\sprob}{\mF}

\newcommand{\dimw}{\xi}

\newcommand{\Hess}[1]{\operatorname{Hess}{#1}}

\newcommand{\tr}{\operatorname{tr}}

\newcommand{\addmodfunc}{\kappa}

\newcommand{\varmargin}{L}

\newcommand{\rev}[1]{#1}            % Clean version (no color)
\begin{document}
\title{Myopically Verifiable Probabilistic Certificates for Safe Control and Learning}
\author{Zhuoyuan Wang*, Haoming Jing*, Christian Kurniawan, Albert Chern, Yorie Nakahira
% \IEEEmembership{Member, IEEE}
\thanks{This work is sponsored in part by the National Science Foundation under grant number 2442948, in part by the Japan Science and Technology Agency under grant number JPMJPR2136, and in part by the Office of Naval Research under grant number N00014-23-1-2252.
The views expressed are those of the authors and do not reflect the official policy or position of the US Navy, Department of Defense or the US Government.}
\thanks{*These authors contributed equally.}
\thanks{Zhuoyuan Wang, Haoming Jing and Yorie Nakahira are with the Department of Electrical and Computer Engineering, Carnegie Mellon University, PA 15213 USA  (e-mail: zhuoyuaw, haomingj, ynakahir@andrew.cmu.edu).}
\thanks{Christian Kurniawan is with the Information Systems Technology and Design pillar, Singapore University of Technology and Design, Singapore 487372 (e-mail: christian.paryoto@gmail.com).}
\thanks{Albert Chern is with the Department of Computer Science and Engineering, University of California San Diego, CA 92093 USA (e-mail: alchern@ucsd.edu).}}

\maketitle

\begin{abstract}
This paper addresses the design of safety certificates for stochastic systems, with a focus on ensuring long-term safety through fast real-time control. In stochastic environments, set invariance-based methods that restrict the probability of risk events in infinitesimal time intervals may exhibit significant long-term risks due to cumulative uncertainties/risks. On the other hand, reachability-based approaches that account for the long-term future may require prohibitive computation in real-time decision making. To overcome this challenge involving stringent long-term safety vs. computation tradeoffs, we first introduce a novel technique termed `probabilistic invariance'. This technique characterizes the invariance conditions of the probability of interest. When the target probability is defined using long-term trajectories, this technique can be used to design myopic conditions/controllers with assured long-term safe probability. Then, we integrate this technique into safe control and learning. The proposed control methods efficiently assure long-term safety using neural networks or model predictive controllers with short outlook horizons. The proposed learning methods can be used to guarantee long-term safety during and after training. Finally, we demonstrate the performance of the proposed techniques in numerical simulations. Code is available at~\href{https://github.com/jacobwang925/MCLS}{https://github.com/jacobwang925/MCLS}.
\end{abstract}


\section{Introduction}
\label{sec:introduction}
\IEEEPARstart{A}{utonomous} control systems must make safe control decisions in real-time and in the presence of uncertainties. For example, autonomous driving cars need to produce specific pedal, brake and steering commands to ensure their safety on road while interacting with other road users.
Many techniques have been developed for deterministic systems with bounded uncertainties that can assure long-term safety with efficient computation~\cite{blanchini1999set, khalil2002nonlinear, ames2019control}. These works often leverage set-invariance-based approaches~\cite{nagumo1942lage,bony1969principe,brezis1970characterization} to transform long-term design specifications into sufficient myopic conditions. 
However, when a system has uncertainty with unbounded support (\eg Gaussian noise), imposing set invariance on the state space for short-term future can no longer guarantee long-term safety. This is because even if the set invariance is satisfied with high probability at each time, uncertainty and risk may still accumulate over time. 

% dean2021guaranteeing is deterministic
Consequently, there are stringent trade-offs between ensuring safety for longer time horizons and reducing computation burden. \rev{For example, robust and stochastic barrier functions can be efficiently evaluated~\cite{clark2021control,9561894}, but are insufficient to ensure long-term guarantees~\cite{so2025comment}.} 
\rev{On the other hand, conditions to control the probability of continuously satisfying set invariance conditions for a long period of time cannot be converted into myopic conditions~\cite{luo2019multi}.} Similarly, reachability-based approaches (e.g., probabilistic reachability) often come with heavy computation to propagate uncertainties over time~\cite{abate2008probabilistic,chapman2019risk,kariotoglou2013approximate,abate2006probabilistic,liao2022probabilistic,vasileva2020probabilistic}. 
These drawbacks pose significant challenges to control and learning methods for stochastic systems operating under latency-critical environments. 
% without overapproximation~\cite{prajna2007framework,abate2006probabilistic}

To address such challenges, in this paper, we propose a novel technique, termed \textit{probabilistic invariance}. 
This technique is inspired by set invariance, but it defines invariance on the probability space, rather than on the state space.
This technique inherits the computational efficiency of the invariance-based approaches and the long-term guarantees of the reachability-based approaches. An infinitesimal future trajectory of a stochastic dynamical system provides little information about its future invariant sets on the state space. In comparison, our key insight is that its probabilistic analog, namely an infinitesimal future value of a long-term probability, does provide explicit computable information about invariant sets in the space of probability values. 

Building upon this insight, we first characterize myopic conditions that can ensure conditions on long-term probability at all times.
%When the probabilistic of interest is associated with long-term trajectories such as probabilistic reachability, this technique inherits the computational efficiency of the invariance-based approach and the long-term guarantees of the reachability-based approach. In stochastic systems, an infinitesimal future trajectory is insufficient to evaluate set invariance on state space for long-term intervals because it requires characterizing the propagation of uncertainties in the entire intervals. In contrast, we show that an infinitesimal future value of a long-term probability is sufficient to evaluate probabilistic invariance, which can assure conditions on long-term probability at all times.   
Specifically, the conditions can be used to constrain two types of long-term probabilities to stay within a desirable range. 
The first type is the probability of forward invariance, which can be used to represent the probabilities of exiting certain sets (e.g., safe regions) within a time interval.
The second type is the probability of forward convergence, which can be used to represent the probabilities of reaching or recovering to certain sets. 
For affine control systems, these conditions can be converted into linear control constraints, and thus can be easily integrated into convex/quadratic programs. 
Then, we apply the probabilistic invariance technique to safe control and learning problems. 
In the context of control, the proposed technique equips the nominal controllers, which are allowed to be neural networks and black-box controllers, with long-term safety. 
In the context of learning, the proposed technique minimally changes the training procedures while ensuring the long-term safety of the control policies during and after training. The contributions of this paper are summarized below: 
\begin{itemize}
    \item We propose a novel technique termed probabilistic invariance, which can be used to find myopic conditions with assured long-term guarantees (Theorem~\ref{lm:main_lemma}).
    \item Building upon these conditions, we present probabilistic safety certificate to be used in real-time control (Section~\ref{S:Proposed_Method}) and test its performance in simulated environments (Section~\ref{sec:experiment_for_control}).
    \item We show how to integrate the probabilistic certificate with policy gradient and Q-learning (Section~\ref{sec:safe_pg_new}) and demonstrate their efficacy in simulation (Section~\ref{sec:experiment_for_rl}).
\end{itemize}

\rev{The rest of the paper is organized as follows. We present related work in Section~\ref{sec:related_work} and preliminaries in Section~\ref{S:Preliminary}. We then present the problem formulation in Section~\ref{S:Problem_Statement}. We introduce the proposed probabilistic safety certificate in Section~\ref{S:Proposed_Method}, where we establish theoretical guarantees for long-term safety and provide practical computation methods for implementation. Afterwards, we present how the proposed safety certificate can be integrated with reinforcement learning in Section~\ref{sec:safe_pg_new}. We present experiment results in Sections~\ref{sec:experiment_for_control} and~\ref{sec:experiment_for_rl} and finally we conclude the paper in Section~\ref{sec:conclusion}.}

\section{Related Work}
\label{sec:related_work}

\subsection*{Safety Certificate for Stochastic Systems}

For deterministic systems with bounded uncertainties, previous work has developed various techniques for safe control (see review articles such as~\cite{ames2019control,blanchini1999set,bansal2017hamilton}), and the combination of set invariance and reachability-based techniques has been studied in~\cite{choi2021robust}. 
As this paper considers stochastic systems, our review below focuses on existing approaches for stochastic systems. 
We approximately classify these approaches into three main types based on their choice of tradeoffs: long-term safety with heavy computation; myopic safety with low computation; and long-term conservative safety with low computation.

Extensive literature considers the distributions of the system state in the long term. 
Reachability-based techniques characterize the probability that states reach or avoid certain regions over a long period of time~\cite{abate2008probabilistic,chapman2019risk,kariotoglou2013approximate,abate2006probabilistic,liao2022probabilistic,vasileva2020probabilistic, vzikelic2023learning}, which are useful to evaluate the long-term safety of actions. 
However, these techniques often come with significant computational costs. 
The cause is two-fold. First, possible trajectories often scale exponentially with the length of the outlook time horizon. Second, rare events are more costly to sample and estimate than nominal events. Compared to these techniques, the proposed method only myopically evaluates the immediate future, and the safety constraints can be given in closed forms for affine control systems. These features help reduce real-time computation.

On the other hand, extensive work has been done to develop efficient controllers in latency-critical stochastic systems. Given the hard tradeoffs between time horizon vs. computation, many of these techniques impose set invariance on short-term or infinitesimal future states. For example, stochastic control barrier functions use a sufficient condition to ensure that the state moves within the tangent cone of the safe set on average~\cite{clark2019control,wang2021safety, santoyo2021barrier, nishimura2024control}. The probabilistic barrier certificate ensures that conditions based on control barrier functions are met with high probability~\cite{luo2019multi,9561894,singletary2022safe, samuelson2018safety}. The myopic nature of these methods achieves a significant reduction in computational cost. However, they can result in unsafe behaviors on a longer time horizon due to the accumulation of tail probabilities of hazardous events. In contrast, the proposed technique provides long-term safe probability guarantees by finding probabilistic invariance conditions that directly control the likelihood of accumulating tail events.

Other approaches, such as the barrier certificate, establish efficient methods to characterize safe action in a given time interval using conservative or approximate conditions. In these studies, probability bounds or martingale approximations are used to obtain sufficient conditions for long-term safety~\cite{prajna2007framework,yaghoubi2020risk,santoyo2021barrier,huang2017probabilistic,anand2019verification}. 
Many of such conditions can be integrated into convex optimization problems to synthesize safe controllers offline or verify control actions online. The controllers using these conditions often require less computation to design or execute.   
However, due to the approximate nature, control actions can be conservative and unnecessarily compromise performance. In contrast, the proposed techniques can use exact safe probability values (computed offline) and inform the probability of exposed risk in the event of infeasibility. 
These features allow control actions to be determined based on accurate probabilities without overconservatisms arising from overapproximation, and the aggressiveness of the performance to be systematically designed based on exposed risks.

\subsection*{Integration with Learning-based Techniques}

Here, we provide a brief summary of related work that integrates safety certificates (safety constraints) with control and reinforcement learning.

The safety certificates mentioned above (\eg \cite{ames2019control,clark2019control,luo2019multi,9561894}) are commonly used to certify or modify nominal controllers. 
The nominal controllers are often allowed to be either black-box models or represented by neural networks~\cite{mazouz2022safety}. 
%These certificates can be used to find safe actions that require minimal changes from nominal controllers or be integrated into optimization-based control. 
As stated above, in stochastic systems, when safety certificates are formulated using myopic conditions, they do not necessarily ensure long-term safety. 
When they are constructed based on approximations or designed to be safe in worst-case uncertainties, the system performance may degrade significantly as the size of uncertainties increases. On the other hand, the proposed certificate can be used similarly to existing safety certificates but ensures long-term safety using myopic controllers and exhibits graceful degradation for increasing uncertainties.\footnote{For example, in its application to extreme driving, the former can be observed from \cite[Fig.~2]{gangadhar2022adaptive}, and the latter from \cite[Fig.~6]{gangadhar2022adaptive}.}
Alternatively, safety considerations can be captured in rewards and constraints in model predictive control (MPC)~\cite{Farina2016,Hewing2020,brudigam2021stochastic}. When safety specifications for future states are directly expressed in the reward and constraint functions, the MPC outlook time horizon needs to cover such future states. 
In comparison, the proposed technique can be used to construct MPC constraints in a way that ensures safety for a longer period than the MPC outlook horizon.

There exists extensive literature on safe reinforcement learning, as reviewed in~\cite{gu2022review, garcia2015comprehensive}. 
Some approaches account for safety specifications in rewards or constraints~\cite{xu2021crpo, chen2021primal, liu2021policy, wachi2020safe, chow2018lyapunov} or design/learn certain functions related to safe conditions~\cite{qin2022neural,wei2022safe,zhao2021learning,tasse2023rosarl,amani2021safe}. Examples of these functions are control barrier functions~\cite{qin2022neural}, safety index~\cite{wei2022safe}, barrier certificate~\cite{zhao2021learning}, rewards and constraints~\cite{tasse2023rosarl,amani2021safe}. 
These methods aim to learn optimal and safe policies but may not guarantee safety in the initial learning phase.
Other approaches impose additional structures on control policies~\cite{donti2020enforcing,srinivasan2020learning,alshiekh2018safe,li2020robust, emam2022safe}. 
For example, safety filters (\eg Lyapunov conditions, control barrier functions, typically defined by prior knowledge of system models) are attached to learning-based controllers, and can often ensure safety or stability for deterministic systems or bounded uncertainties. Similarly, the proposed technique can be used to construct a safety filter for learning-based controllers. When used as a safety filter, it can control, both during and after training, the probability of forward invariance for a given duration or forward convergence within a given horizon for stochastic nonlinear systems.

\section{Preliminary}
\label{S:Preliminary}

Let $\R$, $\R_+$, $\R^n$, and $\R^{m\times n}$ be the set of real numbers, the set of non-negative real numbers, the set of $n$-dimensional real vectors, and the set of $m \times n$ real matrices, respectively. Let $x[k]$ be the $k$-th element of vector $x$. Let $X_{a:b}$ denote the set $\{X_i:i\in\{a,a+1,\cdots,b\}\}$, where $a,b\in\mathbb{Z}$ and $a\leq b$. Let $f:\gX \rightarrow \gY$ represent that $f$ is a mapping from space $\gX$ to space $\gY$. Let \(\mathbb{1}(\gE)\) be an indicator function, which takes \(1\) when condition \(\gE\) holds and \(0\) otherwise. Let $\mathbf{0}_{m\times n}$ be an $m\times n$ matrix with all entries $0$. Given events $\gE$ and $\gE_c$, let $\mP(\gE)$ be the probability of $\gE$ and $\mP(\gE | \gE_c)$ be the conditional probability of $\gE$ given the occurrence of $\gE_c$. Given random variables $X$ and $Y$, let $\mE[X]$ be the expectation of $X$ and $\mE[X | Y = y]$ be the conditional expectation of $X$ given $Y=y$. We use upper-case letters (\eg $Y$) to denote random variables and lower-case letters (\eg $y$) to denote their specific realizations.

\begin{definition}[Infinitesimal Generator]
    The infinitesimal generator $A$ of a stochastic process $\{Y_t\in\R^n\}_{t\in\R_+}$ is
    \begin{align}
    \label{eq:afy}
        AF(y)=\lim_{h\to 0}\frac{\mE[F(Y_h)| Y_0 = y]-F(y)}{h}
    \end{align}
    whose domain is the set of all functions $F:\R^n\rightarrow\R$ such that the limit of \eqref{eq:afy} exists for all $y\in\R^n$.
\end{definition}

\section{Problem Statement}
\label{S:Problem_Statement}

Here, we introduce the control system in subsection~\ref{sec:SystemDescription}, 
% define the measures to characterize two types of safety in subsection~\ref{SS:Characterization of Safe Behaviours}, 
define the safety specification in subsection~\ref{SS:Safety_Specification}, and state the controller design goals in subsection~\ref{SS:Design_Goal}. 

\subsection{Control System Description} 
\label{sec:SystemDescription}
 
We consider a time-invariant control-affine stochastic control and dynamical system. The system dynamics is given by the stochastic differential equation (SDE)
\begin{align}
\label{eq:x_trajectory}
    dX_t = \left(f(X_t) + g(X_t)U_t\right) dt + \sigma(X_t) dW_t,
\end{align}
where $X_t \in \mathcal{X} \subseteq \R^n$ is the system state, $U_t \in \mathcal{U} \subseteq \R^m$ is the control input, and $W_t \in \R^\dimw$ captures the system uncertainties. Here, $X_t$ can include both the controllable states of the system and the uncontrollable environmental variables such as moving obstacles. We assume that $W_t$ is the standard Wiener process with $0$ initial value, \ie $W_0=0$. The value of $\sigma(X_t)$ is determined based on the size of uncertainty in unmodeled dynamics, environmental variables and noise. %We assume that $f:\R^n\rightarrow\R^n$, $g:\R^n\rightarrow\R^{n\times m}$ are locally Lipschitz functions of the state. 

The control action $U_t$ is determined at each time by the control policy. We assume that accurate information of the system state can be used by the control policy. The control policy is composed of a nominal controller and additional modification scheme to ensure the safety specifications illustrated in subsection~\ref{SS:Safety_Specification}. The nominal controller is represented by 
\begin{align}
\label{eq:nominal_controller}
    U_t = N(X_t), 
\end{align}
where function $N$ lies in a class of stochastic or deterministic functions mapping $\mathcal{X}$ to $\mathcal{U}$, \ie $N \in \mathcal N_s\cup\mathcal N_d$. Here, $\mathcal N_s$ denotes the class of stochastic functions, $\mathcal N_d$ denotes the class of deterministic functions, and we use $\mathcal N$ to denote the function class when it can be either.
The design of $N$ does not necessarily account for the safety specifications defined below and can be represented by neural networks. To adhere to the safety specifications, the output of the nominal controller is then modified by another scheme. The overall control policy involving the nominal controller and the modification scheme is represented by 
\begin{align}
\label{eq:generic_controller}
    U_t = N'(X_t, L_t, T_t),
\end{align}
where $N': \R^{n}\times\R\times\R \rightarrow \R^m$ is a mapping from the current state $X_t$, safety margin $L_t$, and time horizon $T_t$ to the current control action $U_t$.  
The formal definition of $L_t$ and $T_t$ will be given later in the section.
The policy of the form~\eqref{eq:generic_controller} assumes that the decision rule is time-invariant.\footnote{\label{ft:timeinvariant-control}The functions $N$ and $N'$ do not change over time.} 
% and that the control action can be uniquely determined for each $(X_t, L_t, T_t)$.
This policy is also assumed to be memory-less in the sense that it does not use the past history of the state $\{X_\tau\}_{\tau < t}$ to produce the control action $U_t$. The assumption for memory-less controller is reasonable because the state evolution $dX_t$ of system \eqref{eq:x_trajectory} only depends on the current system state $X_t$ as $f(X_t)$, $g(X_t)$, and $\sigma(X_t)$ are time-invariant functions of the system state. 
We restrict ourselves to the settings where $f, g,\sigma, N, N’$ have sufficient regularity conditions such that the closed-loop system has a well-defined solution. For example, the system is well-defined when $f, g,\sigma, N, N’$ are bounded and globally Lipschitz~\cite{oksendal2013stochastic}, which are common assumptions usually satisfied in practice. More generally, conditions required to have a unique solution can be found in~\cite[Chapter~1]{oksendal_stochastic_2003a}, \cite[Chapter II.7]{borodin_stochastic_2017} and references therein. 

% We restrict ourselves to the settings where $f, g,\sigma, N, N’$ have sufficient regularity conditions, such that the coupled system of (2) and (4) has an almost surely unique $t$-continuous solution $X_t$. Conditions required to have a unique (strong) solutions can be found in~\cite[Chapter~1]{oksendal_stochastic_2003a}, \cite[Chapter II.7]{borodin_stochastic_2017} and references therein\footnote{The unique solution is a strong solution~\cite{pham2009continuous} as the noise term is explicit in~\eqref{eq:x_trajectory}.}. One set of sufficient conditions is that $f, g,\sigma, N, N’$ are bounded and globally Lipschitz~\cite{oksendal2013stochastic}, which is usually satisfied in practice.

The safe region of the state is specified by the zero super level set of a continuously differentiable barrier function $\phi(x) : \R^n \rightarrow \R$, \ie 
\begin{align}
    \gC(0) = \left\{x \in \R^n : \phi(x) \geq 0 \right\}.
\end{align}
We use 
\begin{align}
\label{eq:l-level_set}
    \gC(L) := \left\{x \in \mathbb{R}^{n}: \phi(x) \geq L \right\}
\end{align}
to denote the $L$-super level set of $\phi$ (the set with safety margin $L$). Accordingly, we use $\gC(L)^c = \{x\in\R^{n}: \phi(x)<L \}$ to denote the unsafe set (complement of the safe set).

% we use $\interior{\gC(L)}  = \{x\in\R^{n}: \phi(x) > L  \}$ to denote the interior of the safe set, $\gC(L)^c = \{x\in\R^{n}: \phi(x)<L \}$ to denote the unsafe set, $\partial \gC(L) = \{x\in\R^n:\phi(x) = L \}$ to denote the boundary of $L$ super level set. 
%We assume that $\phi(x)$ is a second-order differentiable function whose gradient does not vanish at $\partial \gC(0)$.

% Furthermore, we would like to define a nominal controller $N(X_t)$ such that it achieves the desired system performance without considering safety. (move to paragraph before eq 37?)

% \subsection{Probabilistic Characterization of Safe Behaviours}
% \label{SS:Characterization of Safe Behaviours}
\subsection{Safety Specifications}
\label{SS:Safety_Specification}

We consider two types of safety specifications: forward invariance and forward convergence. The two types are used to define the following four types of safety-related probabilities.

% \todo{below are old contents}
%~\cite{chern2021safe}
\subsubsection{Forward Invariance}
\label{SSS:Forward_Invariance}
Long-term safety and long-term avoidance are expressed in the forms of forward invariance conditions.
The forward invariance property refers to the system's ability to keep its state within a set when the state originated from the set. The probabilistic forward invariance to a set $\gC(L_t)$ can be quantified using 
\begin{align}
\label{eq:forward_invariance}
    %\mP\left(\, X_\tau \in \gC(L), \forall \tau \in [0,T] \ | \ X_0=x \, \right)
    \mP\left(\, X_\tau \in \gC(L_t), \forall \tau \in [t,t+T_t] \ | \ X_t=x \, \right)
\end{align}
%\jacob{why are we using the 'big' parentheses? does the second one look better?}
for some time interval $[t,t+T_t]$ conditioned on an initial condition $x \in \gC(L)$.
% \end{remark}

% With slight abuse of notation, we let $T_t := [t, t+T_t]$ to denote the outlook time range when context applies.

% \noindent \ul{Type 1: long-term safety} is defined using the event of forward invariance assuming the continued use of the current (nominal) controller~\eqref{eq:nominal_controller} given $X_t \in \gC(L_t)$ at time $t$: \ie
% \begin{equation}
% \begin{aligned}
%     \label{eq:long_term_safety}
%     & E_1(X_{T_t}, U_{T_t}) := \\
%     & \quad \{X_t \in \gC(L_t), U_\tau = N(X_\tau), X_{\tau} \in \gC(L_t), \forall \tau \in [0,T_t]\}.
% \end{aligned}    
% \end{equation}

% \todo{decide using $L_t$ or $L$ for the definition}
\noindent \ul{Type 1: long-term safety} is defined using forward invariance assuming the continued use of the current (nominal) controller~\eqref{eq:nominal_controller} given $X_t \in \gC(L_t)$ at time $t$. Its probability is given by
\begin{equation}
\begin{aligned}
    \label{eq:long_term_safety}
    & P_1(x, L, T, \phi) := \mathbb{P}\{ X_{\tau} \in \gC(L), \forall \tau \in [t,t+T] \mid X_t = x \},
\end{aligned}
\end{equation}
where the probability is evaluated assuming $U_\tau=N(X_\tau),\forall \tau\in[t,t+T]$.

\noindent \ul{Type 2: long-term avoidance} is defined using forward invariance assuming the future control action can be modified given $X_t  \in \gC(L)$ at time $t$. The long-term avoidance is quantified by
\begin{equation}
\begin{aligned}
    \label{eq:long_term_avoidance}
    & P_2(x, L, T, \phi) := \sup_{N \in \mathcal{N}}\mathbb{P}\{X_{\tau} \in \gC(L), \\
    & \qquad \qquad \forall \tau \in [t,t+T] \mid X_t = x \},
\end{aligned}    
\end{equation}
where the probability is evaluated assuming $U_\tau=N(X_\tau),\forall \tau\in[t,t+T]$. The probability of long-term safety will be used to measure the likelihood of remaining safe under the current controller, and the probability of long-term avoidance will be used to measure the possibility to avoid unsafe events, independently from the current choice of the control policy.

\subsubsection{Forward Convergence}
\label{SSS:Safety_Recovery}
% forward convergence conditions
Finite-time eventuality and reachability are expressed in the forms of forward convergence conditions (\ie having the state converge to a set $\gC$ within $T_t$).
The forward convergence property indicates the system's capability for its state to enter a set when the state originated from outside the set. This probabilistic forward convergence can be quantified using
\begin{align}
\label{eq:forward_convergence}
    %\mP\left(\, \exists \tau \in [0,T] \text{ s.t. } X_\tau \in \gC(L)\ |\ X_0 = x \, \right)
    \mP\left(\, \exists \tau \in [t,t+T_t] \text{ s.t. } X_\tau \in \gC(L_t)\ |\ X_t = x \, \right)
\end{align}
for some time interval $[t,t+T_t]$ conditioned on an initial condition $x \in \gC(L_t)^c $.
% \end{remark}

% \todo{end of old contents}

% \noindent \ul{Type 3: finite-time eventuality} is defined using the event of forward convergence assuming continued use of the current (nominal) controller~\eqref{eq:nominal_controller} where it is given $X_t  \notin \gC(L_t)$ at time $t$: \ie
% \begin{equation}
% \begin{aligned}
%     \label{eq:finite_time_eventuality}
%     & E_3(X_{T_t}, U_{T_t}) := \{X_t \notin \gC(L_t), U_\tau = N(X_\tau), \forall \tau \in [0,T_t], \\
%     & \quad  X_{\tau} \in \gC(L_t), \text{ for some } \tau \in [0,T_t]\}.
% \end{aligned}    
% \end{equation}

\noindent \ul{Type 3: finite-time eventuality} is defined using forward convergence assuming continued use of the current (nominal) controller~\eqref{eq:nominal_controller} where it is given $X_t  \notin \gC(L)$ at time $t$. Its probability is given by
\begin{equation}
\begin{aligned}
    \label{eq:finite_time_eventuality}
    & P_3(x, L, T, \phi) := \mathbb{P} \{X_{\tau} \in \gC(L) \\
    & \quad   \text{ for some } \tau \in [t,t+T] \mid X_t = x \},
\end{aligned}    
\end{equation}
where the probability is evaluated assuming $U_\tau=N(X_\tau),\forall \tau\in[t,t+T]$.
% \begin{align}
%     \label{eq:finite_time_eventuality}
%     U_\tau = N(X_\tau), \forall \tau \in T_t  \;\Longrightarrow\; X_{\tau} \in \gC(L_t) \text{ for some } \tau \in T_t .
% \end{align}

% \noindent \ul{Type 4: finite-time reachability} is defined using the event of forward convergence assuming the future control action can be modified given $X_t  \notin \gC(L_t)$ at time $t$: \ie
% \begin{equation}
% \begin{aligned}
%     \label{eq:finite_time_reachability}
%     & E_4(X_{T_t}, U_{T_t}) := \{X_t \not\in \gC(L_t), \exists\{U_{\tau}\}_{\tau \in [0,T_t]}, \\
%     & \quad  \text{ s.t. } U_\tau \in \mathcal{U}, \forall \tau \in [0,T_t] \textrm{ and } X_{\tau} \in \gC(L_t), \text{ for some } \tau \in [0,T_t]\}.
% \end{aligned}    
% \end{equation}
\noindent \ul{Type 4: finite-time reachability} is defined using forward convergence assuming the future control action can be modified given $X_t  \notin \gC(L_t)$ at time $t$. The finite-time reachability is quantified by
\begin{equation}
\begin{aligned}
    \label{eq:finite_time_reachability}
    & P_4(x, L, T, \phi) := \sup_{N \in \mathcal{N}}\mathbb{P} \{X_{\tau} \in \gC(L)  \\
    & \quad \text{ for some } \tau \in [t,t+T] \mid X_t = x \},
\end{aligned}    
\end{equation}
where the probability is evaluated assuming $U_\tau=N(X_\tau),\forall \tau\in[t,t+T]$. The finite-time eventuality will be used to measure the likelihood of reaching the desired states under the current controller and the finite-time reachability will be used to measure of the system's capability to do so. 

Note that one can define the safety-related probabilities with regard to multiple safety requirements with different horizons of interest. Such generalization is straightforward thus is not explicitly discussed in the paper due to space limit.

\subsection{Design Goals}
\label{SS:Design_Goal}

The objective of this paper is to ensure either long-term safety/avoidance or finite-time eventuality/reachability. The objective is mathematically defined as follows.
For simplicity of notations, we define the following probability
\begin{equation}
\label{eq:safe event}
P(x, L, T, \phi) := 
\begin{cases}
P_1(x, L, T, \phi), \quad \text{for type 1} \\
P_2(x, L, T, \phi), \quad \text{for type 2} \\
P_3(x, L, T, \phi), \quad \text{for type 3} \\
P_4(x, L, T, \phi), \quad \text{for type 4} 
\end{cases}
\end{equation}
The long term safety condition we aim to ensure is defined as
% \begin{equation}
% \label{eq:safety_specification}
% \mP \left(E_{L_t}(X_{T_t}, U_{T_t})\right) \geq 1-\epsilon
% \end{equation}
\begin{equation}
\label{eq:safety_specification}
\mathbb{E}\left[P(X_t, L_t, T_t, \phi)\right] \geq 1-\epsilon, \quad \forall t >0
\end{equation}
% \todo{a function of t, for all t. Done}
% for all $T_t$, conditioned on the initial condition $X_0=x$, 
where $\epsilon \in (0,1)$ is the pre-specified risk tolerance. 
% From now on, all probabilities are conditioned on the initial condition $X_0 = x$ unless otherwise noted. 
The expectation~\eqref{eq:safety_specification} is taken over the distribution of $X_t$ and the future trajectories $\{X_\tau\}_{\tau \in (t, t+T_t]}$ conditioned on the initial distribution of $X_0$.
%In both cases, the probabilities \cref{eq:FI-probability,eq:FC-probability} is taken over the distributions of $X_t$ and its future trajectories $\{X_\tau\}_{\tau \in (t, T_t]}$ conditioned on an initial condition $X_0 = x$. 
The distribution of $X_t$ is generated based on the closed-loop system of \eqref{eq:x_trajectory} and \eqref{eq:generic_controller}, 
% \jacob{only (2), or (2) and (3)? Haoming agrees}
whereas the distribution of $\{X_\tau\}_{\tau \in (t, t+ T_t]}$ is allowed to be defined in two different ways based on the design choice: the closed-loop system of \eqref{eq:x_trajectory} and \eqref{eq:nominal_controller} or the closed-loop system of \eqref{eq:x_trajectory} and \eqref{eq:generic_controller}. 
% \begin{remark}

The physical meaning of the objective~\eqref{eq:safety_specification} is that we want the probability of long-term safety at each time step to be greater than a pre-specified threshold.  
For example, when~\eqref{eq:safe event} is defined for type 1, we have
\begin{equation}
    \mathbb{E}\left[P(X_t, L_t, T_t, \phi)\right] = \mP\left( X_\tau \in \gC(L_t), \forall \tau \in [t,t+T_t] \right).
\end{equation}
% \end{remark}

We consider either fixed time horizon problem or receding time horizon problem, following the definition in ~\cite{rawlings2017model}. In the receding time horizon problem, safety is evaluated at each time $t$ for a time interval $[t,t+H]$. In the fixed time horizon problem, we evaluate, at each time $t$, safety only for the remaining time $[t,H]$. The outlook time horizon for each case is given by
\begin{align}
\label{eq:time_horizon_cases}
    T_t & =
    \begin{cases}
        H, & \text{for receding time horizon,} \\
        H-t, & \text{for fixed time horizon.}
    \end{cases}
\end{align}
The safety margin is assumed to be either fixed or time varying. Fixed margin refers to when the margin remains constant at all time, \ie $\varmargin_t = \ell$. For time-varying margin, we consider the margin $\varmargin_t$ that evolves according to 
% \jacob{$d\varmargin_t / dt$?}
\begin{align}
\label{eq:f_ell_def}
    d\varmargin_t/dt = f_\ell (\varmargin_t),\ \varmargin_0 = \ell,
\end{align}
for some continuously differentiable function $f_\ell$. This representation includes fixed margin by setting $f_\ell (\varmargin_t) \equiv 0$. The values of $T_t $ and $\{\varmargin_t\}_{t \in \mathbb{R}^+}$ are determined based on the design choice.

\section{Probabilistic Safety Certificate}
\label{S:Proposed_Method}

Here, we present a sufficient condition to achieve the safety requirements in subsection \ref{SS:safety_conditions}. Then we provide efficient computation methods to calculate the safety condition in subsection~\ref{SS:safe_prob_computation}. Finally, we propose two safe control strategies and conclude the overall framework in subsection~\ref{SS:Proposed_Algorithm}.
% and outline a method to boost algorithm performance in subsection \ref{SS:Accurate_Gradient}.

Before presenting these results, we first define a few notations. To capture the time-varying nature of $T_t \text{ and } \varmargin_t$, we augment the state space as
\begin{align}
\label{eq:Z_definition}
    Z_t:=[T_t,\varmargin_t,\phi(X_t),X_t^\top]^\top\in\R^{n+3}.
\end{align}
The dynamics of $Z_t$ satisfies the following SDE:
\begin{align}
\label{eq:z_t_dyn}
    dZ_t = (\tilde{f}(Z_t)+\tilde{g}(Z_t)U_t) dt+\tilde{\sigma}(Z_t)dW_t.
\end{align}
Here, $\tilde{f}$, $\tilde{g}$, and $\tilde{\sigma}$ are defined to be
% \begin{align}
%     \tilde{f}(Z_t) & :=\begin{bmatrix}
%     \frac{\partial \phi(X_t)}{\partial t}\\[1ex]
%     \frac{\partial T_t}{\partial t}\\[1ex]
%     \frac{\partial\varmargin_t}{\partial t}\\[1ex]
%     f(X_t)
%     \end{bmatrix} \in \R^{(n+3)}, \\
%     \tilde{g}(Z_t) & :=\begin{bmatrix}
%     \mathbf{0}_{3\times n}\\[1ex]
%     g(X_t)
%     \end{bmatrix} \in \R^{(n+3) \times m}, \\
%     \tilde{\sigma}(Z_t) & :=\begin{bmatrix}
%     \mathbf{0}_{3\times n}\\[1ex]
%     \sigma(X_t)
%     \end{bmatrix} \in \R^{(n+3) \times \dimw}.
% \end{align}
\begin{align}
\label{eq:tilde_f}
    \tilde{f}(Z_t) & :=[f_T,f_\ell (\varmargin_t),f_\phi (X_t),f(X_t)^\top]^\top \in \R^{(n+3)}, \\
    \tilde{g}(Z_t) & :=[\mathbf{0}_{m\times 2},(\gL_{g}\phi(X_t))^\top,g(X_t)^\top]^\top \in \R^{(n+3) \times m}, \\
    \tilde{\sigma}(Z_t) & :=[\mathbf{0}_{\xi\times 2},(\gL_\sigma \phi(X_t))^\top,\sigma(X_t)^\top]^\top \in \R^{(n+3) \times \dimw}.
\end{align}
In \eqref{eq:tilde_f}, the scalar $f_T$ is given by
\begin{align}
\label{eq:f_T_def}
f_T&:=\begin{cases}
        0, & \text{in receding time horizon,} \\
        -1, & \text{in fixed time horizon,}
    \end{cases}
\end{align}
the function $f_\ell$ is given by \eqref{eq:f_ell_def}, and the function $f_\phi$ is given by 
\begin{align}
\label{eq:f_phi_def}
f_\phi (X_t)&:= \gL_{f}\phi(X_t) +\frac{1}{2} \text{tr} \big(\left[\sigma(X_t)\right]\left[\sigma(X_t)\right]^\top\Hess \phi(X_t)\big).
\end{align}

\begin{remark}
\label{rm:lie-derivative}
The Lie derivative of a function $\phi(x)$ along the vector field $f(x)$ is denoted as $\gL_{f}\phi(x) = f(x) \cdot \nabla \phi(x) $. The Lie derivative $(\gL_{g} \phi(x))$ along a matrix field $g(x)$ is interpreted as a row vector such that $\left( \gL_{g}\phi(x) \right) u  = \left( g(x) u \right) \cdot \nabla \phi(x)$.
%Here, the first term $\gL_{\tilde{f}}\sprob(z) = \tilde{f}(z) \cdot \nabla \sprob(z) $ represents the Lie derivative of $\sprob(z)$ along the vector field $\tilde{f}(z)$. In the second term, Lie derivative $(\gL_{\tilde{g}} \sprob(z))$ along a matrix field $\tilde{g}(z)$ is interpreted as a row vector such that $\left( \gL_{\tilde{g}}\sprob(z) \right) u  = \left( \tilde{g}(z) u \right) \cdot \nabla \sprob(z)$.
\end{remark}

% Here, the mapping $D_\phi:\R^{n}\times\R^m \rightarrow \R$ is defined as the infinitesimal generator of the stochastic process $X_t$ acting on the barrier function $\phi$, \ie
% \begin{align}
% \label{eq:infgen_to_phi}
%     \begin{split}
%         D_\phi (X_t,U_t) := &A\phi (X_t) \\
%         = &\gL_{f}\phi(X_t)+\gL_{g}\phi(X_t) U_t \\ 
%         &+\frac{1}{2} \text{tr} \left(\left[\sigma(X_t)\right]\left[\sigma(X_t)\right]^\top\Hess \phi(X_t)\right).
%     \end{split}
% \end{align} 

\subsection{Conditions to Assure Safety}
\label{SS:safety_conditions}

We consider the following probabilistic quantity:\footnote{Recall from Section~\ref{SS:Safety_Specification} that whenever we take the probabilities (and expectations) over paths, we assume that the probabilities are conditioned on the initial condition \(X_0 = x\).}
\begin{align}
\label{eq:cases_summary}
    % \sprob(Z_t) := 
    % \begin{cases}
    %     \mP \left(\minf_{X_t}(T_t) \geq \varmargin_t \right) & \text{for type I,} \\
    %     \mP \left(\exit_{X_t}(\varmargin_t) > T_t \right) & \text{for type II,} \\
    %     \mP \left( \maxf_{X_t}(T_t) \geq \varmargin_t \right) & \text{for type III,} \\
    %     \mP \left( \entrance_{X_t}(\varmargin_t) \leq  T_t \right) & \text{for type IV,}
    % \end{cases}
    % \sprob(Z_t) := \mP(E_{L_t}(X_{T_t}, U_{T_t})),
    \sprob(Z_t) := P(X_t, L_t, T_t, \phi),
\end{align}
where the probability is taken over the same distributions of $\{X_\tau\}_{\tau \in [t, T_t]}$ that are used in the safety requirement~\eqref{eq:safety_specification}. The values of $T_t$ and $\varmargin_t$ (known and deterministic) are defined in \cref{eq:time_horizon_cases,eq:f_ell_def} depending on the design choice of receding/fixed time-horizon and fixed/varying margin. Additionally, we define the mapping $D_\sprob:\R^{n+3}\times\R^m \rightarrow \R$ as\footnote{See Remark \ref{rm:lie-derivative} for the notation for Lie derivative.}
\begin{align}
\label{eq:infgen_to_mf}
    \begin{split}
        D_\sprob (Z_t,U_t) 
        := & \ \gL_{\tilde{f}}\sprob(Z_t)+ \left( \gL_{\tilde{g}}\sprob(Z_t) \right) U_t \\ 
        &+\frac{1}{2}\text{tr} \left(\left[\tilde{\sigma}(Z_t)\right]\left[\tilde{\sigma}(Z_t)\right]^\top\Hess \sprob(Z_t)\right).
    \end{split}
\end{align}
%Here, the first term $\gL_{\tilde{f}}\sprob(z) = \tilde{f}(z) \cdot \nabla \sprob(z) $ represents the Lie derivative of $\sprob(z)$ along the vector field $\tilde{f}(z)$. In the second term, Lie derivative $(\gL_{\tilde{g}} \sprob(z))$ along a matrix field $\tilde{g}(z)$ is interpreted as a row vector such that $\left( \gL_{\tilde{g}}\sprob(z) \right) u  = \left( \tilde{g}(z) u \right) \cdot \nabla \sprob(z)$.
From It\^o's Lemma,\footnote{It\^o's Lemma is stated as below: Given a $n$-dimensional real valued diffusion process $dX = \mu dt + \sigma dW$ and any twice differentiable scalar function $f: \R^n \rightarrow \R$, one has 
$
df= \left(\gL_\mu f + \frac{1}{2}\tr\left(\sigma\sigma^\top \Hess{f}\right)\right) dt + \gL_\mu \sigma dW.
$} the mapping \eqref{eq:infgen_to_mf} essentially evaluates the value of the infinitesimal generator of the stochastic process $Z_t$ acting on $\sprob$: \ie $A\sprob (Z_t) = D_\sprob (Z_t,U_t)$ when the control action $U_t$ is used.

We propose to constrain the control action $U_t$ to satisfy the following condition at all time $t$: 
\begin{align}
\label{eq:safety_condition_each_Zt}
	D_\sprob (Z_t,U_t) \geq -\alpha \left(\sprob(Z_t) - (1-\epsilon) \right).
\end{align}
Here, $\alpha: \R \rightarrow \R$ is assumed to be a monotonically-increasing, concave or linear function that satisfies $\alpha(0) \leq 0$. From \eqref{eq:infgen_to_mf}, condition \eqref{eq:safety_condition_each_Zt} is affine in $U_t$. This property allows us to integrate condition \eqref{eq:safety_condition_each_Zt} into a convex/quadratic program. 
Note that~\eqref{eq:safety_condition_each_Zt} is a forward invariance condition on probability, while typical control barrier function (CBF) based methods perform forward invariance on state space.  
The advantage of imposing forward invariance on the probability space is that long-term safety can be ensured. In contrast, directly imposing forward invariance on the barrier function $\phi$ can not guarantee long-term safety (see Fig.~\ref{fig:worst_case} for comparison results). 
% \begin{enumerate}[leftmargin=\parindent,align=left,labelwidth=\parindent,labelsep=0pt]
%     \item[A1. ] $\alpha$ is a concave or linear function.
%     \item[A2. ] $\alpha$ is monotonically increasing.
%     \item[A3. ] $\alpha(0) \leq 0$.
% \end{enumerate}
\begin{theorem}
\label{lm:main_lemma}
   Consider the closed-loop system of $\eqref{eq:x_trajectory}$ and \eqref{eq:generic_controller}.
   %Assume that $\sprob(z)$ in \cref{eq:cases_summary} is a continuously differentiable function of $z \in \R^{n+3}$ and $\mE [ \sprob(Z_t) ]$ is differentiable in $t$. 
   If system \eqref{eq:x_trajectory} originates at $X_0=x$ with $\sprob(z)>1-\epsilon$, and the control action satisfies \eqref{eq:safety_condition_each_Zt} at all time, then the following condition holds for all time $t \in \R_+$:\footnote{Here, the expectation is taken over $X_t$ conditioned on $X_0 = x$, and $\sprob$ in \cref{eq:cases_summary} gives the probability of forward invariance/convergence of the future trajectories $\{X_\tau\}_{(t , t+T_t]}$ starting at $X_t$.}
	\begin{align}
	\label{eq:satisfy_control_policy}
		\mE\left[\sprob(Z_t) \right] \geq 1 - \epsilon.
	\end{align}
\end{theorem}
\begin{proof}[Proof (\cref{lm:main_lemma})]
First, we show that
\begin{align}
\label{eq:expectation_less_than}
    \mE[\sprob(Z_\tau)] \leq 1-\epsilon
\end{align}
implies
\begin{align}
\label{eq:expectation_morethan_0}
    \mE \left[\alpha\left(\sprob(Z_\tau) - (1-\epsilon) \right) \right] \leq 0,
\end{align}
where we let $\tau$ be the time when \eqref{eq:expectation_less_than} holds. We first define the 
%set of 
events $D_i$ and a few variables $v_i,q_i, \text{ and } \delta_i$, $i\in\{0,1\}$, as follows:
\begin{align}
    \label{eq:D0}
    D_0 & = \left\{\sprob(Z_\tau) < 1-\epsilon \right\},\quad
    D_1 = \left\{\sprob(Z_\tau) \geq 1-\epsilon \right\}, \\
    \label{eq:V0}
    v_0 & = \mE \left[\sprob(Z_\tau) \mid D_0 \right] = 1-\epsilon-\delta_0, \\
    \label{eq:V1}
    v_1 & = \mE \left[\sprob(Z_\tau) \mid D_1 \right] = 1-\epsilon+\delta_1, \\
    \label{eq:P0}
    q_0 & = \mP(D_0), \quad
    q_1 = \mP(D_1).
\end{align}
The left hand side of \cref{eq:expectation_less_than} can then be written as
\begin{align*}
    \mE[\sprob(Z_\tau)] & = \mE \left[\sprob(Z_\tau) \mid D_0 \right] \mP(D_0) +\mE \left[\sprob(Z_\tau) \mid D_1 \right] \mP(D_1) \\
    \label{eq:expectation_in_VP}
    & = v_0q_0 + v_1q_1. \eqnumber
\end{align*}
From
\begin{align}
\label{eq:fact1_from}
    \mE \left[\sprob(Z_\tau) \mid D_0 \right] & < 1-\epsilon\ \text{and}\ \mE \left[\sprob(Z_\tau) \mid D_1 \right] & \geq 1-\epsilon,
\end{align}
we obtain
\begin{align}
\label{eq:fact1}
    \delta_0 \geq 0 \quad\text{and}\quad \delta_1 \geq 0.
\end{align}
Moreover, $\{q_i\}_{i\in\{0,1\}}$ satisfies
\begin{align}
\label{eq:fact2}
    \mP(D_0) + \mP(D_1) = q_0 + q_1 = 1.
\end{align}
Combining \cref{eq:expectation_less_than,eq:expectation_in_VP} gives
\begin{align}
\label{eq:combine_expectation}
    v_0q_0 + v_1q_1 \leq 1-\epsilon.
\end{align}
Applying \cref{eq:V0,eq:V1} to \cref{eq:combine_expectation} gives
\begin{align}
\label{eq:fact3_from}
    \left(1-\epsilon-\delta_0 \right)q_0 + \left(1-\epsilon+\delta_1 \right)q_1 \leq 1-\epsilon,
\end{align}
which, combined with \eqref{eq:fact2}, yields
\begin{align}
\label{eq:fact3}
    \delta_1 q_1 - \delta_0 q_0 \leq 0.
\end{align}
On the other hand, we have
\begin{align*}
    & \mE \left[\alpha\left(\sprob(Z_\tau) - (1-\epsilon) \right) \right] \\
    &\ \ \ \ \ = \ \mP(D_0) \left( \mE \left[\alpha\left(\sprob(Z_\tau) - (1-\epsilon) \right) \mid D_0 \right] \right) \\
    &\ \ \ \ \ \ \ \ + \mP(D_1) \left( \mE \left[\alpha\left(\sprob(Z_\tau) - (1-\epsilon) \right) \mid D_1 \right] \right) \eqnumber \\
    &\ \ \ \ \ = \ q_0 \left( \mE \left[\alpha\left(\sprob(Z_\tau) - (1-\epsilon) \right) \mid D_0 \right] \right) \\
    &\ \ \ \ \ \ \ \ + q_1 \left( \mE \left[\alpha\left(\sprob(Z_\tau) - (1-\epsilon) \right) \mid D_1 \right] \right) \label{eq:expectation_given_q} \eqnumber \\
    &\ \ \ \ \ \leq \ q_0 \left( \alpha\left(\mE \left[\sprob(Z_\tau) - (1-\epsilon) \mid D_0 \right] \right) \right) \\
    &\ \ \ \ \ \ \ \ + q_1 \left( \alpha\left(\mE \left[\sprob(Z_\tau) - (1-\epsilon) \mid D_1 \right] \right) \right) \label{eq:inequality_from_jensen_rule} \eqnumber \\
    &\ \ \ \ \ = \ q_0 \left( \alpha\left(-\delta_0 \right) \right) + q_1 \left( \alpha\left(\delta_1 \right) \right) \label{eq:jense_inequality_given_V0V1} \eqnumber \\
    &\ \ \ \ \ \leq \ \alpha \left(-q_0\delta_0 + q_1\delta_1 \right) \label{eq:jensen_inequality_given_V0V1_assume_A2} \eqnumber \\
    &\ \ \ \ \ \leq \ 0. \label{eq:expectation_lessthan_0} \eqnumber
\end{align*}
Here, \eqref{eq:expectation_given_q} is due to~\eqref{eq:P0};~\eqref{eq:inequality_from_jensen_rule} is obtained from Jensen's inequality~\cite{Jensen1906} for concave function $\alpha$;~\eqref{eq:jense_inequality_given_V0V1} is based on~\eqref{eq:V0} and~\eqref{eq:V1};~\eqref{eq:jensen_inequality_given_V0V1_assume_A2} is given by the assumptions on function $\alpha$; and~\eqref{eq:expectation_lessthan_0} is due to~\eqref{eq:fact3}. Thus, we showed that~\eqref{eq:expectation_less_than} implies~\eqref{eq:expectation_morethan_0}.

Using Dynkin's formula, given a time-invariant control policy, the sequence $\mE[\sprob(Z_t)]$ takes deterministic value over time where the dynamics is given by
\begin{align}
    \frac{d}{d\tau} \mE[\sprob(Z_\tau)] = \mE [A\sprob(Z_\tau)].
\end{align}
Condition~\eqref{eq:safety_condition_each_Zt} implies
\begin{align}
\label{eq:dif_E_bigger_than_E}
    \mE [A\sprob(Z_\tau)] \geq - \mE \left[\alpha\left(\sprob(Z_\tau) - \left(1-\epsilon\right)\right)\right].
\end{align}
Therefore, we have
\begin{align}
\label{eq:dif_E_bigger_than_0}
        \frac{d}{d\tau} \mE [\sprob(Z_\tau)] \geq 0 \quad\text{whenever \(\mE[\sprob(Z_\tau)] \leq 1 - \epsilon\)}.
\end{align}
This condition implies
\begin{align}
    \mE[\sprob(Z_t)] \geq 1-\epsilon  \quad\text{for all \(t\in\R_+\).}
\end{align}
due to Lemma~\ref{lm:lm2}, which is given below.
\end{proof}

\begin{lemma}
\label{lm:lm2} 
Let \(y\colon \R_+\to\R\) be a real-valued differentiable function that satisfies ${d\over dt}y_t \geq 0$ whenever $y_t\leq L$. If $y_0>L$, then $y_t\geq L$ for all $t\in\R_+$.
\end{lemma}
\begin{proof} [Proof (\cref{lm:lm2})]
    Suppose there exists \(b\in\R_+\) such that \(y_b<L\). By the intermediate value theorem, there exists \(a\in (0,b)\) such that \(y_a = L\), and \(y_t<L\) for all \(t\in(a,b]\). Next, by the mean value theorem, there exists \(\tau\in (a,b)\) such that \((dy_t/dt)|_{t=\tau} = (y_b-y_a)/(b-a) < 0\). This contradicts condition that ${d\over dt}y_t \geq 0$ whenever $y_t\leq L$.
\end{proof}

\rev{Theorem~\ref{lm:main_lemma} says that with the proposed safety certificate~\eqref{eq:safety_condition_each_Zt}, long-term safety can be guaranteed with high probability $1-\epsilon$ at all times.}

\subsection{Efficient Computation of Long-term Probabilities}
\label{SS:safe_prob_computation}

% \todo{methods work for both deterministic and stochastic, here present results for deterministic which is more commonly considered}
In this section, we provide computation methods to calculate the probability of the safety-related events.
While the proposed safe control technique works for both stochastic and deterministic control policies, we present here computation methods for the more commonly considered cases of deterministic policies. 
Before introducing computation methods for controller implementation, we first consider a discretized version of dynamical system \eqref{eq:x_trajectory} for sampling time $\Delta t$:
\begin{align}
    \label{eq:discrete_time_dynamics}
    X_{t+\Delta t}=F_d(X_t,U_t,W_t),
\end{align}
where $F_d$ is the discrete-time system dynamics. We also consider a discrete time horizon $T_d=\frac{T}{\Delta t}$.\footnote{Here, we assume that $T$ and $\Delta t$ are chosen such as that $T_d$ is an integer.} For simplicity, with a slight abuse of notation, we use $X_k$ to denote the system state evaluated at time $k\Delta t$ in this subsection.

\subsubsection{Long-term safety and finite time eventuality}
% ~\eqref{eq:long_term_safety}
% \todo{explain}
% The nominal controller is fixed and we can calculate the probability efficiently via p
% ath integral control~\cite{zhang2023optimal, thijssen2015path}. 
% Path integral control and 
Importance sampling techniques provide efficient computation methods to calculate the probabilities when the nominal control is fixed~\cite{zhang2023optimal, thijssen2015path}.\footnote{For simplicity, we assume $g(X) = \sigma(X)$ as in the path integral formulation.} 
Assume $N \in \mathcal{N}_d$ is the nominal controller of interest, and $N_s \in \mathcal{N}_d$ is another controller that we have sampled data with. Let $\mathcal{P}_N$ denote the measure over the classical Wiener space $\Omega = C([0, T ]; \mathbb{R}^n)$ induced by the dynamics~\eqref{eq:discrete_time_dynamics} with $U_k = N(X_k)$, and $\mathcal{P}_{N_s}$ denote the measure associated with the process with $U_k = N_s(X_k)$.
% \footnote{The path space $\Omega$ is a classical Wiener space.} 
Note that $N_s$ can output zero values where the process is uncontrolled. Let $\tilde{U}_k = N(X_k)-N_s(X_k)$ denote the difference between the nominal controller and the sampled controller, where $X_k$ is the state of the closed-loop system with $N_s(X_k)$. 
% given a trajectory $\{X_t, t \in T_t\}$ under controller $U_t = N_s(X_t)$, we can calculate
% \begin{equation}
% \label{eq:IS_weight}
%     \frac{d \mathcal{P}_N}{d \mathcal{P}_{N_s}} = \exp \left\{\int_0^T \frac{1}{2}\frac{\left\|N(X_t) - N_s(X_t)\right\|^2}{\sigma_t^2} d t+\frac{(N(X_t) - N_s(X_t))_t^{\top}}{\sigma_t} d W_t\right\}.
% \end{equation}

Let $S_l = \mathds{1}(X_k \in \mathcal{C}(L), \forall k \in \{1, 2, 3, \cdots , T_d\})$ and $S_f = \mathds{1}(X_k \in \mathcal{C}(L), \text{ for some } k \in \{1, 2, 3, \cdots , T_d\})$ be the indicator function of safety events for long-term safety and finite-time eventuality. We use $S$ to denote $S_l$ or $S_f$ given context.
Through importance sampling, one can sample $\sprob(z)$ based on controller $N_s$ and take weighted expectation with coefficients $\frac{d\mathcal{P}_N}{d\mathcal{P}_{N_s}}$ to estimate $\sprob(z)$ for $N$ as follows\footnote{Recall that $Z_t = [T, L, \phi(X_t), X_t]$ as defined in~\eqref{eq:Z_definition}. Here, we assume that $T$ and $L$ are fixed. The lowercase letters $x$, $z$ are used to denote the specific realizations of $X_t$ and $Z_t$.}
\begin{equation}
\label{eq:importance_sampling}
\begin{aligned}
    % \sprob_N(X) & := \mathbb{E}_{\mathcal{P}_N}[\sprob(X)] 
    % % [E(X_t, U_t), t \in T_t | X_0 = X]
    % = \mathbb{E}_{\mathcal{P}_{N_s}}\left[\sprob(X)\frac{d\mathcal{P}_N}{d\mathcal{P}_{N_s}}\right]. \\
    % % & = \mathbb{E}\left[\mathbb{P}(E(X_{T_t}, N_s(X_{T_t})) | X_0 = X) \frac{d \mathcal{P}_N}{d \mathcal{P}_{N_s}} \right].
     & \sprob(z):= \mathbb{P}\left(S \mid X_0 = x\right) \\
     = \; &\mathbb{E}_{\mathcal{P}_N}[S \mid X_0 = x] 
    = \mathbb{E}_{\mathcal{P}_{N_s}}\left[S \frac{d\mathcal{P}_N}{d\mathcal{P}_{N_s}} \mid X_0 = x \right]. \\
\end{aligned}
\end{equation}
By the Girsanov theorem~\cite{oksendal2013stochastic}, we have
% \begin{equation}
% \label{eq:girsanov_theorem}
%     \frac{d \mathcal{P}_N}{d \mathcal{P}_{N_s}} = \exp \left\{\int_0^T \frac{1}{2}\frac{\left\|\tilde{u}_t\right\|^2}{\sigma_t^2} d t+\frac{\tilde{u}_t^{\top}}{\sigma_t} d W_t\right\}.
% \end{equation}
\begin{equation}
\label{eq:girsanov_theorem}
    \frac{d \mathcal{P}_N}{d \mathcal{P}_{N_s}} = \exp \left\{\sum_{k=0}^T \frac{1}{2} \left\|\tilde{U}_k\right\|^2 d t +\tilde{U}_k^{\top} d W_k\right\}.
\end{equation}
% Here, we regularize $\tilde{u}_t$ with $\sigma_t$, because we assume additive noise on the state instead of additive disturbance on the control in the dynamics~\eqref{eq:x_trajectory}.
With~\eqref{eq:importance_sampling} and~\eqref{eq:girsanov_theorem}, one can estimate the safety probability for any controller $N$ by sampling controller $N_s$. 
% With~\eqref{eq:importance_sampling}, we are able to estimate safety probability of controller $N$ given data on $N_s$. 
Note that with samples on $N_s$, we can estimate $\sprob(z)$ for arbitrary $N$ using~\eqref{eq:importance_sampling}.
Plus, when $N$ and $N_s$ are both represented by neural networks, we can still calculate $\frac{d \mathcal{P}_N}{d \mathcal{P}_{N_s}}$ by sampling $\Tilde{U}_k$ at each time step. 
Specifically, to perform the calculation, one can sample the process with controller $N_s$ for $n_{\text{sample}}$ times. We use 
\begin{equation}
\label{eq:index_long_term_safety}
    s_i := \mathds{1}(X_k \in \mathcal{C}(L), \forall k \in \{1, 2, 3, \cdots , T_d\})
\end{equation}
as an indicator of safety for the $i$-th sampled trajectory for long-term safety probability calculation, and 
\begin{equation}
\label{eq:index_finte_time_eventuality}
    s_i := \mathds{1}(X_k \in \mathcal{C}(L), \text{ for some } k \in \{1, 2, 3, \cdots , T_d\})
\end{equation}
for finite-time eventuality. In this way, we can estimate~\eqref{eq:importance_sampling} with
\begin{equation}
\label{eq:IS_calculation}
    \sprob(z) \approx \frac{1}{n_{\text{sample}}}\sum_{i=1}^{n_{\text{sample}}} s_i w_i,
\end{equation}
where $w_i = \frac{d \mathcal{P}_N}{d \mathcal{P}_{N_s}}$ is the importance sampling weight in~\eqref{eq:girsanov_theorem} for the $i$-th sampled trajectory.
The calculation procedures for long-term safety probability and finite-time eventuality are summarized in Algorithm~\ref{alg:path_integral_IS}.

\begin{algorithm}
\caption{Path Integral Importance Sampling}\label{alg:path_integral_IS} 
\begin{algorithmic}

\State \textbf{Input:} $N$, $N_s$, $T_d$, $n_{\text{sample}}$, $z = [T, L, \phi(x), x]$

\For {$i \in\{1,2,3,\cdots,n_{\text{sample}}\}$}
% \For {$X \in \mathcal{X}$}
\State $X_0 \gets x$
\For {$k\in \{0, 1, 2, \cdots , T_d-1\}$}

\State $U_k \gets N_s(X_k)$
\State $\tilde{U}_k \gets N(X_k)-N_s(X_k)$

\State Calculate $X_{k+1}$ using dynamics \eqref{eq:discrete_time_dynamics} with $U_k$

\EndFor

\State Calculate $s_i$ using~\eqref{eq:index_long_term_safety} or~\eqref{eq:index_finte_time_eventuality}
\State Calculate $w_i$ using~\eqref{eq:girsanov_theorem}
\EndFor
\State \textbf{Output:} Estimate $\sprob(Z)$ using~\eqref{eq:IS_calculation}
% \EndFor

\end{algorithmic}
\end{algorithm}

\begin{remark}
    % For the controlled (or uncontrolled) process with $N_s$,
    Physics-informed learning can also be used to further improve the efficiency of computation for safety-related probabilities in a way that generalizes to unsampled states or time horizons~\cite{wang2023generalizable,wang2024physics,ACC24,wang2025generalizable}.
\end{remark}

% \todo{note that for NN controller we can still sample $N_s$ and $N$ and calculate $\tilde{u}$}
% \jacob{added arguements on NN}

% For 
% finite-time convergence and 
\subsubsection{Long-term avoidance and finite-time reachability}
The probability of long-term avoidance and finite-time reachability can be computed using approximate dynamic programming (DP)~\cite{kariotoglou2013approximate, abate2006probabilistic}. 
% where the computation method in~\cite{kariotoglou2013approximate} can be used for our purposes with some straightforward extension. 
We first show that both the discrete estimation of long-term avoidance and finite-time reachability can be formulated in the form of the reach-avoid problem defined in~\cite{kariotoglou2013approximate}. We define 
% \textcolor{blue}{[Haoming: I am not sure if this way of presenting conditional distribution given a controller is appropriate.]}
\begin{align*}
    \sprob_k(X) := &\sup_{N \in \mathcal{N}_d}\mathbb{P}\{X_{\tau} \in \gC(L),\\
    &\quad\forall \tau \in \{k,k+1,k+2\cdots,T_d\} \mid X_k = X \}\\
    = &\sup_{N \in \mathcal{N}_d}\mathbb{E} \Bigg[\prod_{j=k}^{T_d} \mathbb{1}\left(X_j\in\gC\right)\Bigg| X_k=X\Bigg] \label{eq:long_term_avoidance_equivalence}\eqnumber
\end{align*}
in case of long-term avoidance and 
\begin{align*}
    & \; \sprob_k(X) \\
    := &\sup_{N \in \mathcal{N}_d}\mathbb{P} \{X_{\tau} \in \gC(L) \\
    &\quad\text{ for some } \tau \in \{k,k+1,k+2\cdots,T_d\} \mid X_k = X \}\\
    =& \sup_{N \in \mathcal{N}_d}\mathbb{E}\Bigg[\sum_{j=k}^{T_d}\left(\prod_{i=k}^{j-1} \mathbb{1}\left(X_i\in\gC^c\right)\right) \mathbb{1}\left(X_j\in\gC\right)\Bigg|X_k=X\Bigg],\label{eq:finite_time_reachability_equivalence}\eqnumber
\end{align*}
in case of finite-time reachability.
Here, $\gC^c$ is the complement of the safe set, i.e., unsafe set, and the expectations are taken assuming $U_\tau=N(X_\tau),\forall \tau\in\{k, k+1, k+2 \cdots,T_d\}$. Note that the right hand side of both \eqref{eq:long_term_avoidance_equivalence} and \eqref{eq:finite_time_reachability_equivalence} are special cases of the problem defined in~\cite{kariotoglou2013approximate}.
The goal is to estimate $\sprob_k(X)$ for $X \in \mathcal{X}$ and $k = \{T_d, T_d - 1, T_d-2 \cdots, 0\}$ using DP. Note that the indexing is in reverse time order as we solve for $\sprob_k(X)$ backward in time.

Let $\hat{\sprob}_{k} (X)$ denote the estimation of these quantities at time $k\Delta t$, which estimates \eqref{eq:long_term_avoidance_equivalence} for long-term avoidance or \eqref{eq:finite_time_reachability_equivalence} for finite-time reachability. We define
\begin{equation}   
\mathcal{T}_{U}^{\text{avoid}}[V]\left(X\right) :=  \mathbbm{1}(X \in \gC) \int_{\mathcal{X}} V(Y)Q(dY|X,U)
\end{equation}
for long-term avoidance, and we define
\begin{equation}    \mathcal{T}_{U}^{\text{reach}}[V]\left(X\right) := \mathbbm{1}(X \in \gC) + \mathbbm{1}(X \notin \gC) \int_{\mathcal{X}} V(Y)Q(dY|X,U)
\end{equation}
% for finite time reachability. \textcolor{red}{[Looks like $X\in\mathcal{C}$ is considered reachable, which is different from \eqref{eq:finite_time_eventuality} and \eqref{eq:finite_time_reachability}.]} 
for finite-time reachability.
Here $Q(x'| x, u) := \mP(X_{i+1}=x' \mid X_i=x, U_i=u)$ is the process stochastic kernel describing the evolution of $X$, and $V: \mathcal{X} \rightarrow [0,1]$ is any function of interest~\cite{kariotoglou2013approximate}. For our purpose, we choose $V$ to be $\hat{\sprob}_{k}$. We use $\mathcal{T}_{U}$ to denote $\mathcal{T}_{U}^{\text{avoid}}$ or $\mathcal{T}_{U}^{\text{reach}}$ given context for simplicity.

We define the linear programming as
\begin{equation}
\label{eq:LP}
\begin{aligned}
\min _{w_1, \ldots, w_M} & \int_{\mathcal{X}} \hat{\sprob}_k(X) dX \\
\text { s.t } & \hat{\sprob}_k\left(X^s\right) \geq \mathcal{T}_{U^s}\left[\hat{\sprob}_{k+1}\right]\left(X^s\right), \quad \forall s \in\left\{1, \ldots, n_s\right\} \\
& \hat{\sprob}_k\left(X^s\right)=\sum_{i=1}^M w_i \prod_{j=1}^n \psi\left(X^s[j]; c_{i, j}, \nu_{i, j}\right)
\end{aligned}
\end{equation}
where $(X^s, U^s)$ are sampled from a uniform distribution on the space $\mathcal{X}\times\mathcal{U}$, and $n_s$ is the number of samples. Here, we use Gaussian kernels to parameterize $\hat{\sprob}_k$, where $M$ is the number of kernels, $n$ is the number of features, 
and
% is the parameterization of the probability, and
\begin{equation}
    \psi\left(x ; c, \nu\right):=\frac{1}{\sqrt{2 \nu \pi}} e^{-\frac{1}{2} \frac{\left(x-c\right)^2}{\nu}}
\end{equation}
is the Gaussian function. The constants $c_{i,j}$, $i\in\{1,2,3,\cdots,M\}$, $j\in\{1,2,3,\cdots,n\}$ are sampled from a uniform distribution in $\gC$ (or a bounded subset of interest in $\gC$) for long-term avoidance and in $\gC^c$ (or a bounded subset of interest in $\gC^c$) for finite-time reachability. The constants $\nu_{i,j}$, $i\in\{1,2,3\cdots,M\}$, $j\in\{1,2,3,\cdots,n\}$ are randomly sampled on $\mathbb{R}^+$.\footnote{In practice, we can uniformly sample $v_{i, j} \sim \text{Unif}(1,2)$. As the Gaussian Radial basis functions are universal approximators with sufficiently many basis elements~\cite{park1991universal}, $\hat{\sprob}_{0}(X)$ can approximate the true value with arbitrary precision as the number of basis elements ($M$ and $n$) increase.} We solve the weight $w_i$ for the linear program~\eqref{eq:LP} to approximate $\hat{\sprob}_k(X)$. 
% \todo{elaborate}
The entire recursive approximation process is summarized in Algorithm~\ref{alg:value_approx}. The output $\hat{\sprob}_{0}(X)$ is an approximation of $\sprob(Z)$ in terms of long-term avoidance and finite-time reachability.

\begin{algorithm}
\caption{Recursive Approximation}\label{alg:value_approx} 
\begin{algorithmic}

\State \textbf{Input:} $T_d$

\For {$i\in\{1,2,3,\cdots,M\}$}
\For {$j\in\{1,2,3,\cdots,n\}$}
\State Sample $c_{i,j}$ from a uniform distribution in $\gC$ or $\gC^c$
\State Sample $\nu_{i,j}$ on $\mathbb{R}^+$
% \textcolor{blue}{[some distribution on $(0,\infty)$]}
\EndFor
\EndFor

\State Initialize $\hat{\sprob}_{T_d}(X) = \mathbbm{1}(X \in \gC)$ for all $X \in \mathcal{X}$

\For {$k \in \{T_d-1,T_d-2,T_d-3 \cdots, 0\}$}
\For {$s\in\{1,2,3,\cdots,n_s\}$}
% \in \mathcal{X} \times \mathcal{U}$}
\State Sample $(X^s, U^s)$ from a uniform distribution on $\mathcal{X} \times \mathcal{U}$
\State Evaluate $\mathcal{T}_{U^s}[\hat{\sprob}_{k+1}](X^s)$

% \State Evaluate $b(X_s, U_s) = \mathcal{T}_{U^s}[\hat{\sprob}_{k+1}](X_s)$

\EndFor

\State Solve the linear program~\eqref{eq:LP} to get $\hat{\sprob}_{k}(X)$

\EndFor

\State \textbf{Outputs:} $\hat{\sprob}_{0}(X)$

\end{algorithmic}
\end{algorithm}

\begin{remark}
% With the numerical calculation described above, the computed safety probability can have approximation errors. Even though there is no analytical error bounds for the proposed computation scheme in terms of discretization interval or number of samples, the estimation is unbiased and consistent, \ie with large enough number of samples and small enough interval, the estimate is accurate. Besides, one can find the long-term safety probability of a stochastic system by solving a corresponding PDE~\cite{chern2021safecontrol}, and numerical errors can be evaluated through off-the-shell solvers. Once the error bound or an estimate of the error bound is obtained, one can reduce the risk tolerance $\epsilon$ accordingly to ensure the desired long-term safety probability. 
Based on~\cite{thijssen2015path, hammersley2013monte}, with a large enough number of samples, the estimation error can be made arbitrarily small with a sufficient number of samples. Besides, the long-term safety probability of Type 1 and Type 3 are solutions to partial differential equations (PDEs)~\cite{chern2021safe, prandini2006stochastic}. There exist numerical methods to solve these PDEs such as the Crank–Nicolson method~\cite{crank1947practical} that are known to be convergent, \ie the numerical solution approaches the exact solution as the grid size approaches $0$. 
\end{remark}

\subsection{Safe Control Algorithms}
\label{SS:Proposed_Algorithm}

Here, we propose two safe control schemes based on the safety conditions introduced in subsection~\ref{SS:safety_conditions}. In both schemes, the value of $\sprob$ is defined in~\eqref{eq:safe event} as type 1 or 2 when the safety specification is given as forward invariance condition, and as type 3 or 4 when the safety specification is given as forward convergence condition.

\subsubsection{Additive modification}
\label{SSS:additive}

We propose a control policy of the form 
\begin{align}
\label{eq:additive_modification_policy}
    N'(X_t, L_t, T_t)=N(X_t)+\addmodfunc(Z_t)(\gL_{\tilde{g}}\sprob(Z_t))^\top.
\end{align}
Here, $N$ is the nominal control policy defined in \eqref{eq:nominal_controller}. 
The mapping $\addmodfunc:\R^{n+3}\rightarrow\R_+$ is chosen to be a non-negative function that are designed to satisfy the assumptions of Theorem~\ref{lm:main_lemma} and makes $U_t = N'(X_t, L_t, T_t)$ to satisfy \eqref{eq:safety_condition_each_Zt} at all time. Then, the control action $U_t = N'(X_t, L_t, T_t)$ yields
\begin{align}
\label{eq:additive_modification_policy-exp}
% &\mE [ d\sprob(Z_t) ] = 
% & A \sprob(Z_t)= 
&D_\sprob (Z_t,N'(X_t,L_t,T_t)) = \gL_{\tilde{f}}\sprob(Z_t) + (\gL_{\tilde{g}}\sprob(Z_t)) N(X_t) \nonumber \\
& \quad +\addmodfunc \gL_{\tilde{g}}\sprob (Z_t)\left(\gL_{\tilde{g}}\sprob (Z_t)\right)^\top +
\frac{1}{2}\text{tr} \left( \tilde{\sigma}\tilde{\sigma}^\top\Hess \sprob (Z_t)\right).
\end{align}
As $\addmodfunc$ is non-negative, the term $\addmodfunc \gL_{\tilde{g}}\sprob \left(\gL_{\tilde{g}}\sprob \right)^\top $ in \eqref{eq:additive_modification_policy-exp} takes non-negative values. This implies that the second term additively modify the nominal controller output $N(X_t)$ in the ascending direction of the safety probability. 

% \todo{empirical risk minimization for NN to design $\kappa$}
One can find the optimal $\kappa$ through the following constrained empirical risk minimization~\cite{marcus2023constrained}
\begin{align}
\label{eq:find_k}
\begin{aligned}
    \kappa^* = && \argmin_\kappa & \ \  \mathbb{E}\left[J(\kappa, Z_t)\right] \\
    && \text{s.t.   }  \ \
    & D_\sprob (Z_t,N'(X_t,L_t,T_t)) \\
    && &+\alpha \left(\sprob(Z_t) - (1-\epsilon) \right) \geq 0,
\end{aligned}
\end{align}
where $J(\kappa,Z_t)$ is the objective function of choice to minimize. The expectations in \eqref{eq:find_k} are taken over the process of the closed-loop system with~\eqref{eq:additive_modification_policy}. 
\rev{Note that~\eqref{eq:find_k} only needs to be solved once offline, then the safe control can be calculated via~\eqref{eq:additive_modification_policy} online. In practice, through polynomial parameterization (or other basis functions parameterization) of $\kappa$,~\eqref{eq:find_k} can be effectively solved with sum-of-squares optimizations~\cite{parrilo2000structured, lasserre2009moments}, similar to what has been done in the barrier certificate literature~\cite{prajna2007framework}. }
% \todo{why expectation? expected shall be removed}

\subsubsection{Constrained optimization}
\label{SSS:conditioning}
% \todo{add MPC setting}
Similarly, we can formulate the constrained optimization in the model predictive control (MPC) setting as the following
\begin{equation}
\label{eq:safe_mpc_formulation}
\begin{aligned}
    & \argmin_{\mathbf{u} = \{U_t, \cdots , U_{t + (\bar{H}-1)\Delta t}\}}  \ 
    \mathbb{E} \left[J(X_{t+\Delta t}, X_{t+ 2\Delta t}, \cdots, X_{t+\bar{H} \Delta t}, \mathbf{u})\right] \\
    & \qquad \text{s.t.} \ \ D_\sprob (Z_t,U_t) \geq -\alpha \left(\sprob(Z_t) - (1-\epsilon) \right), \\
    & \qquad \qquad X_{\tau+\Delta t}=F_d(X_\tau,U_\tau,W_\tau),\\
    & \qquad \qquad \qquad \qquad \forall \tau\in\{t,t+\Delta t,\cdots,t + (\bar{H}-1)\Delta t\},\\
\end{aligned}
\end{equation}
where $J$ is the objective function of MPC, $\bar{H}$ is the outlook prediction horizon of MPC, \rev{and $D_\sprob (Z_t,U_t)$ is the infinitesimal generator defined in~\eqref{eq:infgen_to_mf}}. The control action $N'(X_t,L_t,T_t)$ takes the value of $U_t$ in $\mathbf{u}$.
The optimization problem is designed to satisfy the assumptions of Theorem~\ref{lm:main_lemma} to comply with the safety specification~\eqref{eq:safety_specification}.
Note that once we formulate the safe control problem into constrained optimization, one can add additional constraints into the same optimization problem~\cite{garcia1989model,mayne2000constrained}. 
It is worth noting that the MPC prediction horizon $\bar{H}$ can be shorter than the horizon to ensure safety because the safety constraint is only imposed for $Z_t$ in \eqref{eq:safe_mpc_formulation}. 
% When $\left( \gL_{\tilde{g}}\sprob(z) \right) \neq 0$ for any $z$, there always exists $u$ that satisfy the constraint \eqref{eq:safety_condition_each_Zt}. 
% \textcolor{red}{[A bit confused by this sentence.]} 
% \jacob{if we have control constraint e.g. $u \in \mathcal{U}$, the safety condition might not always be feasible}
%\textcolor{red}{todo: state when it cannot be satisfied }
%\textcolor{red}{todo: state it's computation}
%\textcolor{red}{todo: explain F can be computed offline}
In the special case where the MPC prediction horizon is 1, the constrained optimization becomes
\begin{align}
\label{eq:conditioning}
    \begin{aligned}
        N'(X_t, L_t, T_t) = && \argmin_{u} & \ \ \mathbb{E} \left[J(N(X_t),u)\right] \\
        && \text{s.t.} & \ \ \eqref{eq:safety_condition_each_Zt},
    \end{aligned}
\end{align}
where $J:\R^m\times\R^m\rightarrow\R$ is the objective function to be minimized. 
Note that one can use a performance-oriented nominal controller to define the objective function where $J$ penalizes the deviation from the nominal controller.
The constraint of \eqref{eq:conditioning} imposes that \eqref{eq:safety_condition_each_Zt} holds at all time $t$, and can additionally capture other design restrictions.\footnote{For example, $N'$ is Lipschitz continuous when $J(N(x),u) = u^\top H(x) u$ with $ H(x)$ being a positive definite matrix (pointwise in $x$).}

Both additive modification and constrained optimization are commonly used in the safe control of deterministic systems (see~\cite[subsection II-B]{chern2021safe} and references therein). These existing methods are designed to find control actions so that the vector field of the state does not point outside of the safe set around its boundary. In other words, the value of the barrier function will be non-decreasing in the infinitesimal future outlook time horizon whenever the state is close to the boundary of the safe set. 
However, such myopic decision-making may not account for the fact that different directions of the tangent cone of the safe set may lead to vastly different long-term safety. In contrast, the proposed control policies \eqref{eq:additive_modification_policy} and \eqref{eq:conditioning} account for the long-term safe probability in $\sprob$, and are guaranteed to steer the state toward the direction with non-decreasing long-term safe probability when the tolerable long-term unsafe probability is about to be violated. 
% When $\sprob$ is defined based on the closed-loop system involving \eqref{eq:x_trajectory} and \eqref{eq:nominal_controller}, its value can be computed offline. In such cases, 
Note that $\sprob$ can be computed offline and the controller only needs to myopically evaluate~\eqref{eq:additive_modification_policy} or solve~\eqref{eq:conditioning} in real-time execution. In this case, the online computation efficiency is comparable to common myopic barrier function-based methods in a deterministic system.

The overall safe control strategy is shown in Algorithm~\ref{alg:safe_control}. Note that we present the algorithm in discrete time as most digital controller requires discretization.
We first initialize $H$ in~\eqref{eq:time_horizon_cases} for outlook time horizon, time step $\Delta t$, initial state $X_0$, nominal controller $N$, and maximum simulation time $T_{\text{max}}$. 
Then we acquire the safety probability of each state using the computation methods introduced in section~\ref{SS:safe_prob_computation}.
At each time step $t$, 
% \textcolor{red}{[when did the system becomes discrete time?]} \jacob{good point. I think for real implementation we will need discretize time anyway? but not sure how to better present this. maybe we can point out in algorithm we need discretization?} \textcolor{red}{[I agree.]} 
% we use techniques introduced in section~\ref{SS:safe_prob_computation} to calculate the safety probability given the current state $X_t$. 
we use the safe control methods in section~\ref{SS:Proposed_Algorithm} to calculate the safe control action $U_t$. We execute this control action to step the dynamics and this together forms our overall safe control strategy.

% \jacob{to discuss the algorithm}

% \begin{algorithm}
% \caption{Safe control}\label{alg:safe_control} 
% % \begin{algorithmic}

% \State Initialize $H$, $\Delta t$, $X_0$, $N$, $T$

% \State $t \gets 0$

% \While{$t < T$}

% \State Calculate $\sprob(X_t)$ through~\eqref{eq:importance_sampling} or~\eqref{eq:LP}

% \State Solve~\eqref{eq:additive_modification_policy} or~\eqref{eq:conditioning} to get safe control $U_t$

% \State Solve the dynamics~\eqref{eq:x_trajectory} with $U_t$ to get $X_{t+\Delta t}$

% \State $t \gets t + \Delta t$

% \EndWhile

% % \end{algorithmic}
% \end{algorithm}

\begin{algorithm}
\caption{Safe control}\label{alg:safe_control} 
\begin{algorithmic}

\State Initialize $H$, $\Delta t$, $X_0$, $N$, $T_{\text{max}}$

Get $\sprob(X)$ for $X \in \mathcal{X}$ through Algorithm~\ref{alg:path_integral_IS} or Algorithm~\ref{alg:value_approx}
% ~\eqref{eq:importance_sampling} or~\eqref{eq:LP}

\State $t \gets 0$

\While{$t < T_{\text{max}}$}

% \State Retrieve $\sprob(X_t)$ 

\State Solve~\eqref{eq:additive_modification_policy} or~\eqref{eq:conditioning} or~\eqref{eq:safe_mpc_formulation} to get safe control $U_t$

\State Step the dynamics~\eqref{eq:x_trajectory} with $U_t$ to get $X_{t+\Delta t}$

\State $t \gets t + \Delta t$

\EndWhile

\end{algorithmic}
\end{algorithm}

\begin{remark}
    In our proposed safe controller, since we only need the value of $N(X_t)$ at time step $t$, we can just query from the neural network and get the value, and there is no need to derive the nominal control from the exact expression of $N(X)$. This feature indicates that our proposed safe control method is naturally compatible with neural network-based nominal controller, and can be used for cases when complex black-box controllers are used.
    % , and can be used for many real-world applications where the nominal controller is complicated or even black-boxed.
\end{remark}

So far, we consider safe control problems and aim to find a constraint on action to be imposed on some nominal controller or in model predictive control. In the following section, we consider reinforcement learning and we aim to learn an optimal controller for a given objective function that also satisfies the safety constraints during and after learning.

\section{Integration with Reinforcement Learning}
\label{sec:safe_pg_new}
In this section, we discuss how to integrate our proposed safe control method with reinforcement learning (RL) methods for control problems. \rev{Specifically, we show how to optimize the control policy with RL in a way that accounts for the potential modification from the safety certificate, to meet RL objectives while ensuring long-term safety.}
% The application of the proposed approach to reinforcement learning requires certain modifications in the algorithms that perform reinforcement learning. 
In the proposed RL methods that incorporate the proposed safety constraints, its performance guarantee (Theorem~\ref{lm:main_lemma}) also applies.

Specifically, we consider policy gradient and Q-learning, as they are the two major approaches for solving reinforcement learning problems~\cite{sutton1999policy, watkins1992q}.
We follow the discrete-time setup for reinforcement learning as described in~\cite[Chapter 3]{sutton2018reinforcement}, and use $\{X_t\}$ and $\{U_t\}$ for state and action instead for notation consistency. 
\rev{Note that in this setup, the ability to perform roll-outs of the system dynamics is assumed, which is essential for calculating values such as safety probabilities and value functions.}

We define the nominal policy parameterized by $\theta$ as\footnote{If one is to use discretized state or action, the notation $\mathbb{P}( \cdot )$ and $\mathbb{P}( \cdot | \cdots )$ is used to denote probability and conditional probability. If one is to use a continuous state or action, the notation they are used to denote density or conditional density functions.}
\begin{equation}
\label{eq:parameterized_policy}
    \pi_\theta(\hat{U}_t \mid X_t) = \mathbb{P}(\hat{U}_t \mid X_t),
\end{equation}
where $X_t$ is the state at time $t$, and $\hat{U}_t$ is the nominal control action at time $t$. Then, we can find the safe control via the proposed probabilistic safety certificate
\begin{equation}
    U_t = \text{Proj}_{\mathcal{S}}(\hat{U}_t),
\end{equation}
where $\text{Proj}_{(\cdot)}$ is operator of taking projection onto $(\cdot)$, and $\mathcal{S}:=\{U : D_\sprob (Z_t,U) \geq -\alpha \left(\sprob(Z_t) - (1-\epsilon) \right) \}$ is the set of safe control actions. We represent this process by  
\begin{equation}
\label{eq:safety_filter}
    G(U_t \mid X_t, \hat{U}_t) = \mathbb{P}(U_t \mid X_t, \hat{U}_t),
\end{equation}
where $G$ outputs a distribution over the action space with inputs of the state and the action sampled from policy $\pi_\theta$. 
% Specifically, we calculate $G$ by solving the following optimization
% % similar to~\eqref{eq:conditioning},
% \begin{align}
% \label{eq:safety_filter_def}
%     \begin{aligned}
%         G(U_t \mid X_t, \hat{U}_t) := \ & \argmin_{\mathcal{N}} \ \ \mathbb{E}_{U \sim \mathcal{N}(X_t)}[J(\hat{U}_t, U)] \\
%      & \text{  s.t.  }  
%         D_\sprob (X_t,U) \geq -\alpha \left(\sprob(X_t) - (1-\epsilon) \right),
%         % \ \ \eqref{eq:safety_condition_each_Zt},
%     \end{aligned}
% \end{align}
% where $\mathcal{N}$ is the class of all possible policies, and $\sprob$ is defined over different types of safety specifications as introduced in section~\ref{SS:Safety_Specification}.
% Note that for any of the specifications (type 1-4), 
We consider $G$ to be independent from the policy parameters $\theta$. For type 2 and type 4 safety specifications, since we take the maximum of the probability over all policies, the safety filter will be independent from $\theta$ by definition. If type 1 and type 3 safety specifications are to be used, then the nominal policy that determines~\eqref{eq:long_term_safety} and~\eqref{eq:finite_time_eventuality} should not depend on $\theta$.

% Here we consider $F$ to be acquired through running the proposed method~\eqref{eq:conditioning} over the safest-possible nominal policy $\bar{\pi}$, \ie
% \begin{equation}
%     \bar{\pi} = \argmax_{\pi} \; \mathbb{P}\left(x_t \in \gC(L_t), u_t \sim \pi(x_t), \forall t \in \{1,\cdots, H\}\right),
% \end{equation}
% which by definition is fixed and independent from $\pi_\theta$.

\subsection{Integration with Policy Gradient}
We first consider integration with policy gradient methods. 
For each episode, we acquire the following sample trajectory
\begin{equation}
    \tau=\{X_{0:H_r}, \hat{U}_{0:H_r-1}, U_{0:H_r-1}\},
\end{equation}
where $H_r$ is the horizon for the finite time reinforcement learning problem.
We define the reward function for each state $X$ and action $U$ to be $r(X, U)$, and thus the cumulative reward for a trajectory $\tau$ becomes
\begin{equation}
    R(\tau)=\sum_{t=0}^{H_r-1} r\left(X_t, U_t\right).
\end{equation}
We define the value function for policy $\pi_\theta$ as
\begin{equation}
    V^\theta=\mathbb{E}_{\tau \sim (\pi_\theta, G)}[R(\tau)],
\end{equation}
which is the expected cumulative rewards under policy $\pi_\theta$. \rev{In the following lemma, we derive the policy gradient calculation for RL with the proposed safety certificate.}

\begin{lemma}
\label{lem:policy_gradient_safety_filter}
Consider the parameterized policy defined in~\eqref{eq:parameterized_policy} and the safety filter represented by~\eqref{eq:safety_filter}, where $G$ does not depend on the policy parameters $\theta$. Then, the policy gradient with the safety filter can be calculated as
\begin{equation}
\label{eq:safe_PG_update}
\begin{aligned}
\nabla V^\theta & = \mathbb{E}_{\tau \sim (\pi_\theta, G)}\left[R(\tau) \sum_{t=0}^{H_r-1} \nabla \log(\pi_\theta (\hat{U}_t \mid X_t)) \right],
\end{aligned}
\end{equation}
where the expectation is taken over sample trajectories generated with policy $\pi_\theta$ together with the safety filter $G$.
\end{lemma}
\begin{proof} [Proof (\cref{lem:policy_gradient_safety_filter})]
% \todo{change $i$ to $t$}
For a trajectory $\tau$, we have
% \begin{equation}
% \begin{aligned}
% \mathbb{P}^\theta(\tau)=\mathbb{P}\left(x_0\right) \prod_{t=1}^{H-1} & \pi_\theta \left(\hat{u}_t \mid x_t\right) \\ 
% & \quad F\left(u_t \mid x_t, \hat{u}_t\right)  \mathbf{P}\left(x_{t+1} \mid x_t, u_t\right),
% \end{aligned}
% \end{equation}
\begin{equation}
\begin{aligned}
\mathbb{P}^\theta(\tau)=\mathbb{P}(X_0) \prod_{t=0}^{H_r-1} \pi_\theta (\hat{U}_t | X_t) G(U_t | X_t, \hat{U}_t)  \mathbb{P}\left(X_{t+1} | X_t, U_t\right).
\end{aligned}
\end{equation}
% where $\mathbf{P}\left(X_{t+1} \mid X_t, U_t\right)$ is the transition kernel from time step $t$ to $t+1$. \todo{define this at the start of the section based on the dynamics equation} $\mathbf{P}\left(x' | x, u\right)$
Then we have the policy gradient being
% \begin{equation}
%     \nabla V^\theta & =\nabla \int R(\tau) d \mathbb{P}^\theta(\tau) =\nabla \int R(\bar \tau) d \mathbb{P}^\theta(\bar \tau) 
% \end{equation}
\begin{equation}
    \nabla V^\theta =\nabla \int_{\mathcal{I}} R(\tau) \mathbb{P}^\theta(\tau)d\tau =\nabla \int_{\bar{\mathcal{I}}} R(\bar \tau) \mathbb{P}^\theta(\bar \tau) d \bar \tau
\end{equation}
where $\mathcal{I}$ is the set of all possible trajectories and $\bar{\mathcal{I}}$ is set of trajectories with non-zero probabilities. Then we have for any $\bar \tau \in \bar{\mathcal{I}}$,
% \begin{equation}
% \label{eq:PG_int}
% \begin{aligned}
% \nabla V^\theta & =\nabla \int R(\bar \tau) d \mathbb{P}^\theta(\bar \tau) \\
% & =\int R(\bar \tau) \nabla d \mathbb{P}^\theta(\bar \tau) \\
% & =\int R(\bar \tau) d \mathbb{P}^\theta(\bar \tau) \nabla \log d \mathbb{P}^\theta(\bar \tau),
% % & =\int R(\tau)\left(\sum_t \log \pi_\theta\left(\hat{u}_t \mid x_t\right)\right) d P \theta_{(\tau)} \\
% % & =\mathbb{E}\left[R(\tau) \sum_t \log \pi \theta\left(\hat{u_t} \mid x_t\right)\right]
% \end{aligned}
% \end{equation}
\begin{align}
\nabla V^\theta & =\nabla \int_{\bar{\mathcal{I}}} R(\bar \tau)\mathbb{P}^\theta(\bar \tau)d\bar\tau=\int_{\bar{\mathcal{I}}} R(\bar \tau) \nabla  \mathbb{P}^\theta(\bar \tau) d \bar\tau \\
& =\int_{\bar{\mathcal{I}}} R(\bar \tau) \mathbb{P}^\theta(\bar \tau) \nabla \log \mathbb{P}^\theta(\bar \tau) d\bar\tau,\label{eq:PG_int}
\end{align}
where we know that
\begin{align}
    \log \mathbb{P}^\theta(\bar \tau)
    & = \nabla \log \biggl[\mathbb{P}(X_0) \prod_{t=0}^{H_r-1} \pi_\theta(\hat{U}_t \mid X_t) G(U_t \mid X_t, \hat{U}_t) \nonumber\\
    & \qquad \qquad \qquad  \mathbb{P}(X_{t+1} \mid X_t, U_t)\biggr] \\
    & = \nabla \biggl[\log(\mathbb{P}\left(X_0\right)) + \sum_{t=0}^{H_r-1} \biggl(\log(\pi_\theta (\hat{U}_t \mid X_t)) \nonumber\\ 
    & + \log(G(U_t \mid X_t, \hat{U}_t)) + \log(\mathbb{P}\left(X_{t+1} \mid X_t, U_t\right))\biggr) \biggr] \\
    & = \nabla \left[\sum_{t=0}^{H_r-1} \log(\pi_\theta (\hat{U}_t \mid X_t)) \right] \label{eq:g_independent}\\
    & = \sum_{t=0}^{H_r-1} \nabla \log(\pi_\theta (\hat{U}_t \mid X_t)).\label{eq:log_p_computation}
\end{align}
Here, \eqref{eq:g_independent} holds because $G$ is independent from the policy $\pi_\theta$.
Combining~\eqref{eq:PG_int} and~\eqref{eq:log_p_computation}, we get~\eqref{eq:safe_PG_update}.
\end{proof}

\rev{Lemma~\ref{lem:policy_gradient_safety_filter} states that the policy gradients with long-term safety constraints can be computed similarly to standard policy gradients, but with modified control actions and policy distributions.} With~\eqref{eq:safe_PG_update}, one can perform policy optimization $\theta \gets \theta + \eta \nabla V^{\theta}$ to train the reinforcement learning agent with $\eta$ being the learning rate~\cite{sutton2018reinforcement}.
We conclude the procedures to apply the proposed safety filter to policy gradient in Algorithm~\ref{alg:safe_pg}. For each policy gradient iteration in the total $N_{\text{iter}}$ iterations, we collect sample trajectories for $N_{\text{eps}}$ episodes, and calculate the modified policy gradient using~\eqref{eq:safe_PG_update}. 
% \textcolor{red}{[(minor) $N$ conflicts with nominal controller $N$]} 
Then we can perform policy optimization with the modified policy gradient. 
\begin{algorithm}
\caption{Safe Policy Gradient}
\label{alg:safe_pg} 
\begin{algorithmic}

\State Initialize $\theta_0$, $G$, $N_{\text{iter}}$, $N_{\text{eps}}$

\For{$i \in\{ 1,2,3,\cdots,N_{\text{iter}}\}$} 
% \Comment{For each iteration}

% \State $k \gets 0$

\For{$k\in\{1,2,3,\cdots,N_{\text{eps}}\}$}

\State Sample trajectory $\tau_k$ with $\pi_{\theta_i}$
\State Calculate $R(\tau_k)$

\EndFor

\State Estimate $\nabla V^{\theta_i}$ in~\eqref{eq:safe_PG_update} with $\{R(\tau_k):k\in\{1,\cdots,N_{\text{eps}}\}\}$
% \State Calculate $\nabla V^{\theta_i}$ in~\eqref{eq:safe_PG_update} with $R(\tau)$ and $\nabla \log (\pi_{\theta_i}(\tau))$ 

\State Perform policy optimization $\theta_{i+1} = \theta_{i} + \eta_i \nabla V^{\theta_i}$
% with learning rate $\eta$

\EndFor

\end{algorithmic}
\end{algorithm}

\subsection{Integration with Q-learning}

% \todo{claim that's tabular setting and how we define control}
In this section, we introduce integration of the safety certificate with Q-learning. For simplicity, we consider the tabular setting with discrete and finite state space and control space as in~\cite{watkins1992q}. Generalization to continuous state and action spaces can be achieved by using deep Q learning based method~\cite{gaskett1999q}.
% \todo{elaborate and put an algorithm}
One can think of the safety filter as part of the environment (open-loop dynamics), so that the Q-learning process will not be affected. We define the modified Q function as
\begin{equation}
\label{eq:modified_Q_function}
    Q^{\pi_{\theta},G}(x, u):=\mathbb{E}_{\tau \sim (\pi_\theta, G)}\left[\sum_{t=0}^{\infty} \gamma^t r(X_t, U_t) \mid X_0=x, U_0=u\right],
\end{equation}
which is the expected discounted reward starting from state $x$ and action $u$, following policy $\pi_\theta$ modified by a safety filter $G$. The only difference with the standard Q function is that~\eqref{eq:modified_Q_function} includes the safety filter as part of the policy. Similarly, we can define the modified Bellman operator as
\begin{equation}
\begin{aligned}
    & \mathcal{T}(Q)(x, u)\\ 
    := \ & r(x, u)+\gamma \;  
 {\mathbb{E}}_{{x^{\prime} \sim \mathbb{P}(X_{t+1} \mid X_t=x, U_t=u)}}\left[\max _{u^{\prime} \in \mathcal{U}} Q\left(x^{\prime}, u^{\prime}\right)\right].
\end{aligned}
\end{equation}
% \todo{change $t$ to $i$}
Then, the Q-learning update rule for each iteration $i$ is given by
\begin{equation}
\label{eq:q_function_update}
\begin{aligned}
& \quad \; Q_{i+1}(x, u) \\
& =\left(1-\eta_i\right) Q_i(x, u)+\eta_i \mathcal{T}\left(Q_i\right)(x, u) \\
& =Q_i(x, u)+\eta_i\left(r(x, u)+\gamma \max _{u^{\prime}} Q_i\left(x^{\prime}, u^{\prime}\right)-Q_i(x, u)\right)
\end{aligned}  
\end{equation}
where $\eta_i$ is the learning rate.
From~\cite{watkins1992q}, when the reward function $r$ is bounded, the learning rate satisfies $0 \leq \eta_i < 1$, and 
\begin{equation}
    \sum_{i=1}^{\infty} \eta_i=\infty, \quad \sum_{i=1}^{\infty} \eta_i^2<\infty,
\end{equation}
the Q-learning converges to the optimal Q function asymptotically with probability 1.

The procedures for the proposed safe Q-learning method are summarized in Algorithm~\ref{alg:safe_q_learning}. Note that the algorithm only learns Q function. There are many choices for how to obtain the control policy using Q function (\eg $\epsilon$-greedy).

\begin{algorithm}
\caption{Safe Q-learning}
\label{alg:safe_q_learning} 
\begin{algorithmic}

\State Initialize $G$, $N_{\text{iter}}$, $N_{\text{eps}}$

\State Initialize $Q_0(x,u)$ for $x \in \mathcal{X}$ and $u \in \mathcal{U}$

\For{$i\in \{0,1,2,\cdots,N_{\text{iter}}\}$} 

\State Initialize $X_0$, $\eta_i$

% \State $Q_{i} \gets Q$

\For{$k\in\{0,1,2,\cdots,N_{\text{eps}}\}$}

\State $U_k \gets \argmax_U Q_i(X_k, U)$

\State Take $U_k$ and observe $X_{k+1}$, $r(X_k, U_k)$

% \State Update $Q_i(X_k, U_k) \gets Q_{i+1}(X_k, U_k)$ using~\eqref{eq:q_function_update}
\State Update $Q_{i+1}(X_k,U_k)$ using $\eqref{eq:q_function_update}$

\EndFor

% \State Update $Q_i(X_k, U_k) \gets Q_{i+1}(X_k, U_k)$ using~\eqref{eq:q_function_update}

% \State $Q\gets Q_i$

\EndFor

\end{algorithmic}
\end{algorithm}

\section{Experiments for Control}
\label{sec:experiment_for_control}
In this section, we show the efficacy of our proposed safe control method in example use cases.

\subsection{Settings}
We compare our proposed controller with three existing safe controllers: stochastic control barrier functions (StoCBF)~\cite{clark2019control}, probabilistic safety barrier certificates (PrSBC)~\cite{luo2019multi}, and conditional-value-at-risk barrier functions (CVaR) \cite{ahmadi2020risk}.
These baseline methods for comparison are widely acknowledged methods for safe control of stochastic systems, and have been used in safe trajectory synthesis and optimization~\cite{pereira2021safe}, safe control of quadrotor swarms~\cite{batra2022decentralized}, risk-sensitive motion planning~\cite{dixit2021risk} and robot locomotion~\cite{schneider2024learning}, \textit{etc}.
Note that all three methods for comparison assume the observation of the current state and the access to the system uncertainty model, which are the basic assumptions in our method as well. 

We consider the following two settings: 
\begin{itemize}[leftmargin=*]
    \item {Worst-case safe control:} We use the controller that satisfies the safety condition with equality at all times to test the safety enforcement power of these safety constraints. Such control actions are the riskiest actions that are allowed by the safety condition. The use of such control actions allows us to evaluate the safety conditions separated from the impact of the nominal controllers. Here we want to see whether our proposed controller can achieve non-decreasing expected safety as intended.
    \item {Switching control:} We impose the safe controller only when the nominal controller does not satisfy the safety constraint. Here we want to see how the proposed controller performs in practical use, where typically there is a control goal that is conflicting with safety requirements.
\end{itemize}
% We run simulations with $dt = 0.1$ for all controllers, and parameters for our proposed controller ($\epsilon = 0.3$), Clark's ($\alpha = 9$), PrSBC ($\epsilon = 0.33$) and CVaR ($\alpha = 0.5, \beta = 0.03$). Each Monte Carlo approximation uses $10000$ sampled trajectories. 

We run simulations with $\Delta t = 0.1$ for all controllers unless otherwise specified. The initial state is set to $x_0 = 3$. For our controller, each Monte Carlo approximation uses $10000$ sampled trajectories. 
Details on the methods for comparison and the used parameters can be found in Appendix~\ref{sec:algs_for_comparison} in the extended version of the paper submitted along with the supplementary materials.

% - if safety condition is really safe (worst-case) 
% - expected behaviors in normal use (switching case) 

\subsection{Example Use Cases}
\subsubsection{System 1}
\label{sec:exp_linear_sys}
We consider the control affine system (\ref{eq:x_trajectory}) with $f(X_t) \equiv A = 2$, $g(X_t) \equiv 1$, $\sigma(X_t) \equiv 2$. 
The safe set is defined as
\begin{align}
\label{eq:safe_set_experiment}
\begin{split}
    \gC(0) = \left\{x \in \R : \phi(x) \geq 0 \right\}, 
\end{split}
\end{align}
with the barrier function $\phi(x) := x-1$. The safety specification is given as the forward invariance condition. 
The nominal controller is a proportional controller $N(X_t) = -2.5 X_t$. The closed-loop system with this controller has an equilibrium at $x=0$ and tends to move into the unsafe set in the state space.

Fig.~\ref{fig:worst_case} shows the results in the worst-case setting. The proposed controller can keep the expected safe probability $\mE[\sprob(X_t)]$ close to $0.9$ all the time, while others fail to keep it at a high level with used parameters. A major cause of failure is the accumulation of rare event probability, leading to unsafe behaviors. This shows the power of having a provable performance for non-decreasing long-term safe probability over time. For comparable parameters, the safety improves from StoCBF to PrSBC to CVaR. This is also expected as constraining the expectation has little control of higher moments, and constraining the tail is not as strong as constraining the tail and the mean values of the tail. 

\begin{figure}
\centering
\includegraphics[width=1\linewidth]{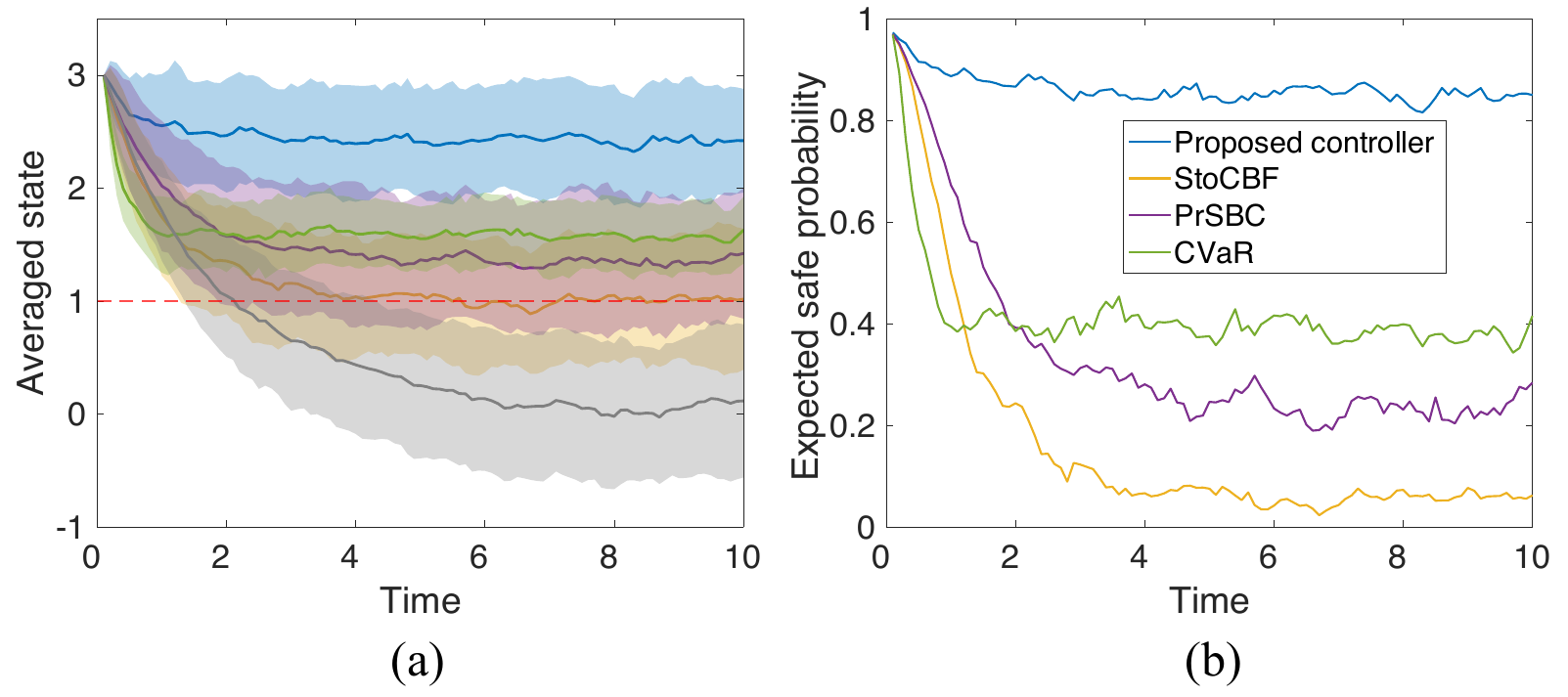}
\caption{Results in the worst-case setting. (a) the average system state over 100 trajectories. Red dotted line indicates the boundary of the safe set. (b) the expected safe probability.
\label{fig:worst_case}}
\end{figure}

Fig.~\ref{fig:switch control} shows the results in the switching control setting. We obtained the empirical safe probability by calculating the number of safe trajectories over all total trials. In this setting, the proposed controller can keep the state within the safe region with the highest probability compared to other methods, even when there is a nominal control that acts against safety criteria. This is because the proposed controller directly manipulates dynamically evolving state distributions to guarantee non-decreasing safe probability when the tolerable unsafe probability is about to be violated, as opposed to when the state is close to an unsafe region. Our novel use of forward invariance condition on the safe probability allows a myopic controller to achieve long-term safe probability, which cannot be guaranteed by any myopic controller that directly imposes forward invariance on the safe set.

\begin{figure}
\centering
\includegraphics[width=1\linewidth]{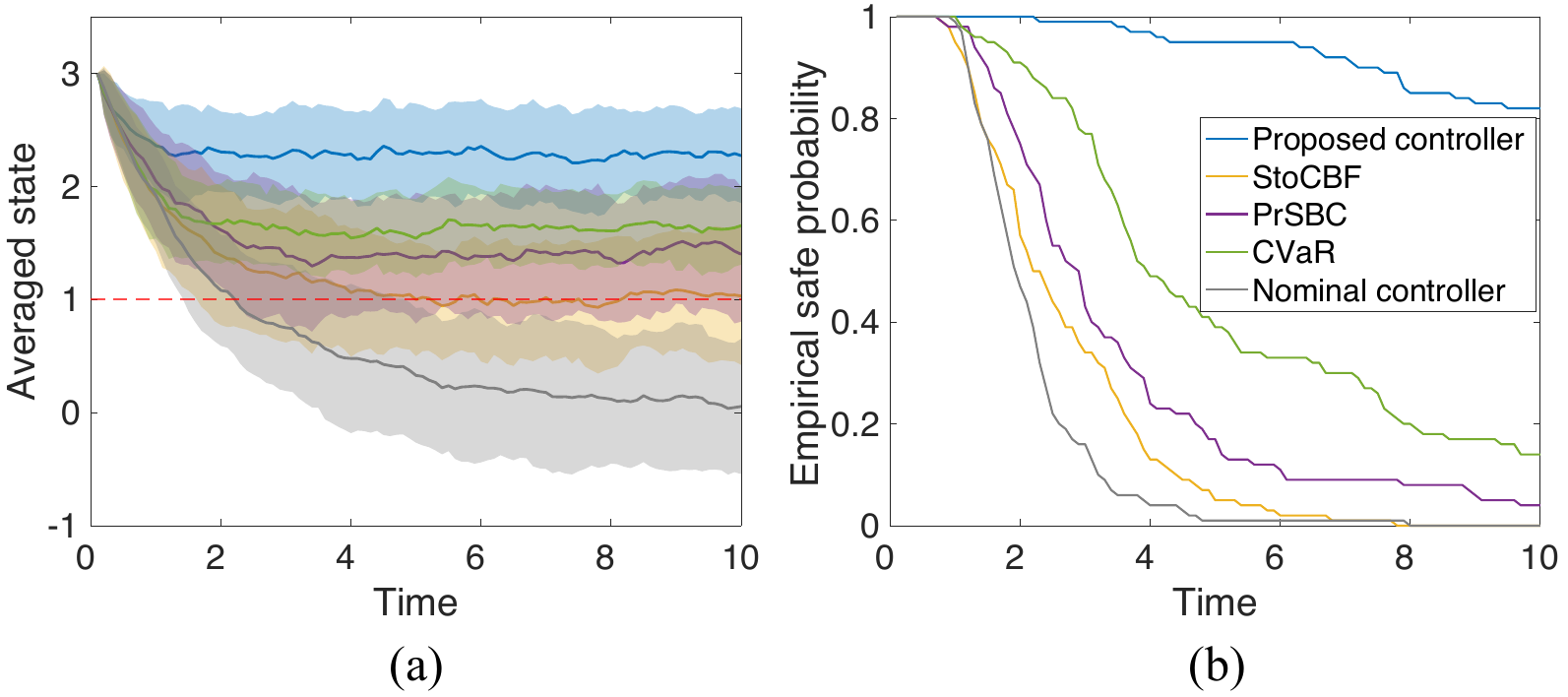}
\caption{Results in the switching control setting. (a) the averaged system state of 100 trajectories with its standard deviation. Red dotted line indicates the boundary of the safe set. (b) the empirical safe probability.
\label{fig:switch control}}
\end{figure}

\subsubsection{System 2}
We modify the system dynamics used in the previous section to show the performance of the proposed controller against nonlinear traps. Specifically, we consider the control affine system (\ref{eq:x_trajectory}) with 
\begin{equation}
\label{eq:nonlinear_dynamics_experiment}
\begin{aligned}
    f(x) = \begin{cases}
        2, \text{ if } x > 1.5 \\
        -3, \text{ otherwise,} 
    \end{cases}
     g(x) = \begin{cases}
        1, \text{ if } x > 1.5 \\
        0, \text{ otherwise.} 
    \end{cases}
\end{aligned} 
\end{equation}
This modification makes the system uncontrollable when the state $x$ reaches or is below $1.5$, even though $x=1.5$ remains safe. We keep noise magnitude $\sigma(X_t) \equiv 2$, safe set $\gC(0) = \left\{x \in \R : x-1 \geq 0 \right\}$, nominal controller $N(X_t) = -2.5 X_t$ the same. With the nominal controller, the system will reach $x=1.5$. Once the system reaches $x=1.5$, the state will to the origin, which is unsafe. The switch in the dynamics and the nonlinear trap will get the system to unsafe regions even though the starting point is safe. This design will help illustrate the advantage of the predictive nature of our proposed controller compared to the existing myopic ones.

Fig.~\ref{fig:worst_case_nonlinear} shows the results in the worst-case setting. The proposed controller can keep the expected safe probability $\mE[\sprob(X_t)]$ close to $0.8$ all the time, while all other methods get $0$ expected safe probability due to the nonlinear trap in the system. This shows the advantage of considering long-term safety instead of safety in the immediate next step. 

Fig.~\ref{fig:switch_control_nonlinear} shows the results in the switching control setting. In this setting, the proposed controller can keep the state within the safe region with high probability, while all other methods get to unsafe regions eventually. This is because the proposed controller encodes future safety information into the long-term safety probability, and avoids the potential traps in the system by enforcing high safety probability at all times.
The use of forward invariance condition on the safe probability also simplifies the control design process, as the proposed controller only needs encoded safety probability and does not require specific system dynamics in the design phase.

Safe control results on neural network-based nominal controllers and 3-dimensional systems can be found 
% in the extended version of the paper~\cite{wang2024myopically}.
in Appendix~\ref{sec:nn_nominal_exp} and Appendix~\ref{sec:3d_results}, respectively.

\begin{figure}
\centering
\includegraphics[width=1\linewidth]{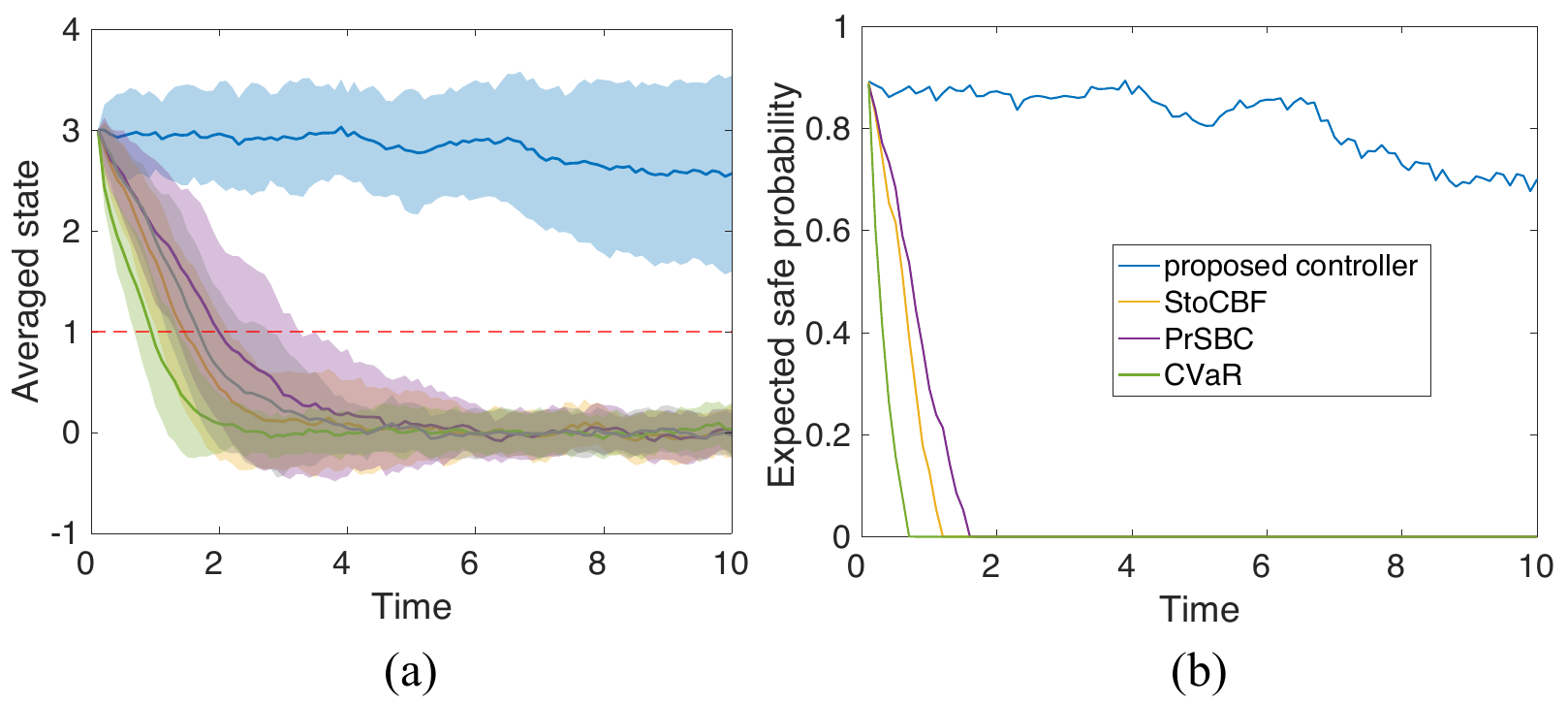}
\caption{Results in the worst-case setting with nonlinear dynamics~\eqref{eq:nonlinear_dynamics_experiment}. (a) the average system state over 50 trajectories. Red dotted line indicates the boundary of the safe set. Black dotted line indicates the boundary of the nonlinear trap. (b) the expected safe probability. 
% \todo{add sub-captions}
\label{fig:worst_case_nonlinear}}
\end{figure}

\begin{figure}
\centering
\includegraphics[width=1\linewidth]{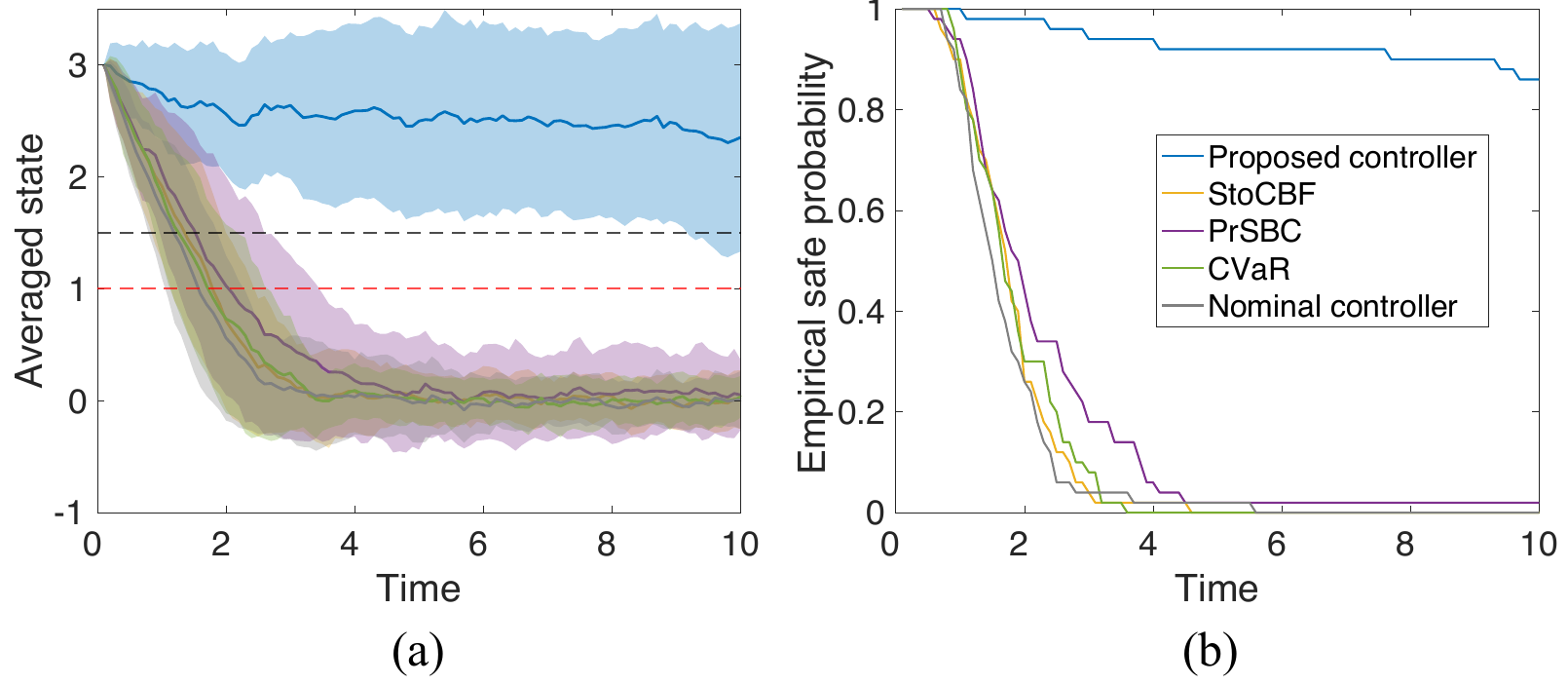}
\caption{Results in the switching control setting with nonlinear dynamics~\eqref{eq:nonlinear_dynamics_experiment}. (a) the averaged system state of 50 trajectories with its standard deviation. Red dotted line indicates the boundary of the safe set. Black dotted line indicates the boundary of the nonlinear trap. (b) the empirical safe probability. 
% \todo{add sub-captions}
\label{fig:switch_control_nonlinear}}
\end{figure}

% \subsection{Integration with Learning-based Controller}

% \todo{add NN experiments}

% \subsection{Computation Analysis}
% We provide the computation time of the proposed method with model predictive control (MPC) based nominal controller in our previous work~\cite{gangadhar2022adaptive}, to show the efficiency of safety certificate. An autonomous driving scenario with safety criteria on tire friction is considered, and the constrained optimization formulation of the safety certificate is used. Fig.~\ref{fig:MPC_computation_safety} shows the computation time and the safety probability with regard to the MPC prediction horizon. It can be seen that the MPC computation time scales roughly linearly with the prediction horizon, while the safety probability is always maintained at a high level with the proposed method. The results indicates that the proposed method can safeguard controllers like MPC without introducing additional computation.

% \begin{figure}
% \centering
% \includegraphics[width=1\linewidth]{Figures/MPC_computation_safety.pdf}
% \caption{Computation load and safety v.s MPC prediction horizons.
% \label{fig:MPC_computation_safety}}
% \end{figure}

% \subsection{Integration with Reinforcement Learning}
\section{Experiments for Reinforcement Learning}
\label{sec:experiment_for_rl}
% Following the derivation in section~\ref{sec:safe_pg_new}, we conduct the following numerical experiments to verify the proposed method.
In this section, we show experiment results of the proposed safety certificate with reinforcement learning.

We consider system~\eqref{eq:x_trajectory} with $f(X_t) \equiv 0$, $g(X_t) \equiv 1$, $\sigma(X_t) \equiv 0.4$. We set $\mathcal{X} = \{0, 1, 2, \cdots 10\}$ and $\mathcal{U} = \{-1,0,1\}$. The discretized system dynamics with $\Delta t=1$ becomes 
% \textcolor{red}{[This implicitly says $dt=1$?]}\jacob{correct. Added that} \textcolor{red}{[Then you probably don't want to say all controllers has $dt=0.1$ in VI.B.]}
% We consider the discrete-time discrete state stochastic system 
% \todo{use the same dynamics and refer to that.}
% \todo{use gaussian noise to fit the proof}
\begin{equation}
    X_{t+1} = X_t + U_t + W_t,
\end{equation}
where $X \in \mathcal{X}$ is the state, $U \in \mathcal{U}$ is the control, and $W \in \mathcal{U}$ is the disturbance. 
% We set $\mathcal{X} = \{0, 1, 2, \cdots 10\}$, $\mathcal{U} = \{-1,0,1\}$. 
The softmax policy $\pi$ is parameterized by $\theta \in \mathbb{R}$ as
\begin{equation}
    \pi_\theta(U \mid X) \propto \exp (\theta X U).
\end{equation}
The reward function $r(x,u): \mathbb{R}^2 \rightarrow \mathbb{R}$ is set to be
\begin{equation}
    r(X, U) = (X-5)/10.
\end{equation}
This reward function will encourage visitation of larger $X$.
Throughout this section we consider the long-term avoidance safety specification (type 2), where the safety filter $F$ can be acquired through running the proposed method over the following safest-possible nominal policy $\bar{\pi}$, \ie
\begin{equation}
    \bar{\pi} = \argmax_{\pi} \; \mathbb{P}\left(X_t \in \gC(L_t), \forall t \in \{1,\cdots, H\}\right)
\end{equation}
where the probability is evaluated assuming $U_t\sim\pi(x_t),\forall t\in \{1,2,3,\cdots,H-1\}$.
% which by definition is fixed and independent from $\pi_\theta$.

\subsection{Policy Gradient}
We consider a finite horizon problem with horizon length $H =10$ and initial state $x_0 = 0$. 
We first run policy gradient without safety filter for 1000 iterations, with 10 episodes per iteration. The policy is updated by
\begin{equation}
    \theta_{i+1}=\theta_i+\eta_i \nabla_\theta V^{\pi_{\theta_i}}(x_0),
\end{equation}
where $\eta_i$ is the learning rate and is set to be $\eta_i = \frac{1}{\sqrt{i}}$ with $i$ being the iteration. We run policy gradient for 1500 iterations with 10 episodes per iteration. The results are shown in Fig~\ref{fig:pg} in blue lines. We can see that the average reward keeps increasing during training. The learned policy will drive the system's state to its maximum value to gain better rewards.
% Red line in Fig.~\ref{fig:theta_pg} shows the change of $\theta$ with regard to policy gradient iterations. We can see that $\theta$ keeps increasing, which is expected as this will gain better cumulative reward.
% As a result, the state trajectory will go up quickly to gain more reward with the learned $\theta$. A sample trajectory with $\theta^{(T)}$ is shown in Fig.~\ref{fig:state_pg} \todo{in red}.

We then apply the proposed safety filter to the system. We consider the safe set to be $\mathcal{C}=\{x\in \R: x-5 > 0\}$, \ie we want to limit the state transition to any $x > 5$ for safety concerns. The safety filter corresponding to the safest nominal policy $\bar{\pi}$ is then given by
% \begin{equation}
%     f(x,u) = \begin{cases}
%     -1, \quad & x > 6 \\
%     0, \quad & 4 < x \leq 6 \\
%     u, \quad & \text{otherwise}
%     \end{cases}
% \end{equation}
\begin{equation}
\label{eq:saefty_filter}
    G(X,U) = \begin{cases}
    -1, \quad & X > 5 \\
    U, \quad & \text{otherwise}
    \end{cases}
\end{equation}
Here the safety filter outputs action `$-1$' when $X>5$, otherwise outputs the nominal action $U$.
This safety filter~\eqref{eq:safety_filter} can be acquired by running the proposed safe control strategy with $\bar{\pi}$ being the nominal controller.
We conduct modified policy gradient described in~\eqref{eq:safe_PG_update} and the results are shown in Fig.~\ref{fig:pg} in orange lines. We can see that the learned policy will maximize the return, but will limit the state to safe regions thanks to the safety filter.
\begin{figure}
    \centering
    \includegraphics[width=1\linewidth]{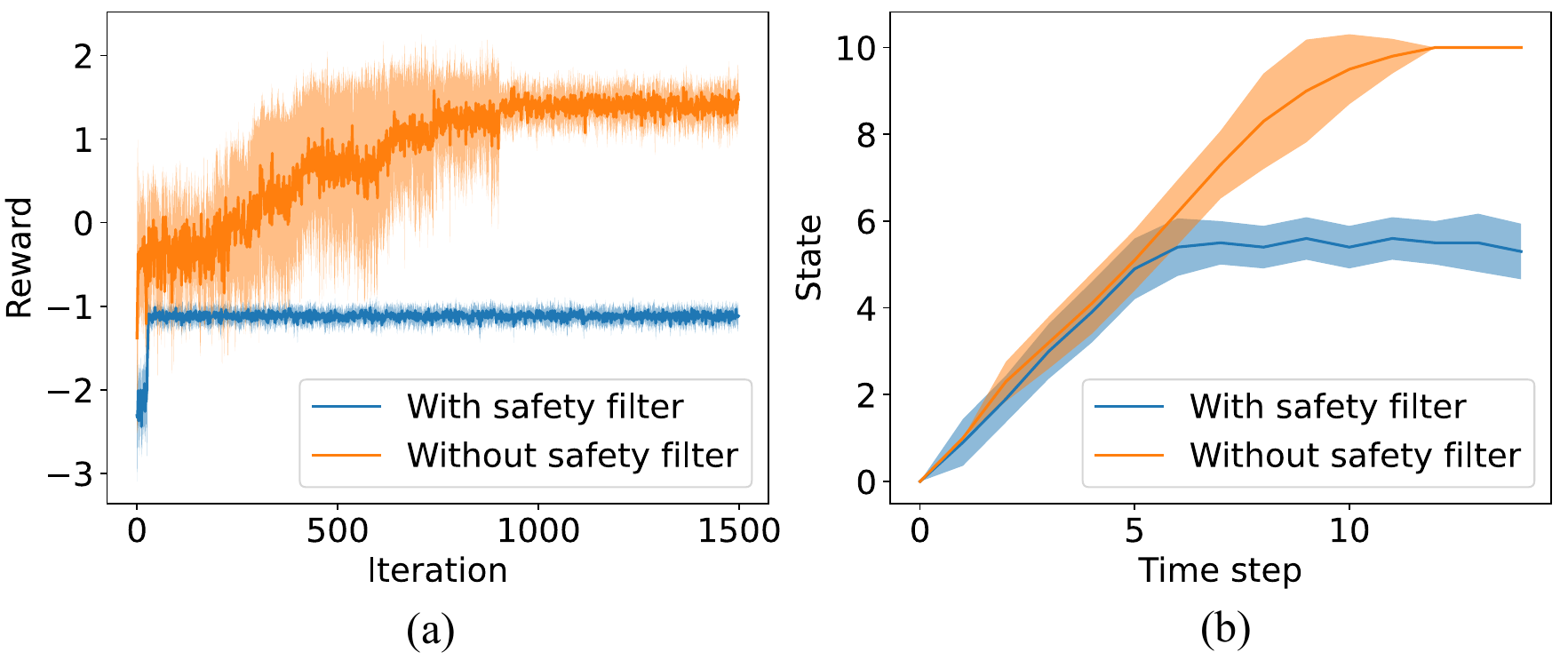}
    \caption{Policy gradient with and without the proposed safety filter. (a) the averaged rewards for 1500 iterations. (b) sample state trajectories with the learned policy. 
    % \todo{change font size in plots}
    }
    \label{fig:pg}
\end{figure}

\subsection{Q-learning} 
We consider the infinite time horizon problem so that Q function only depends on state and action. We set the discount function to be $\gamma = 0.9$ so that the effective horizon is $10$. We set the initial state as $X_0 = 0$ and number of episodes $N_{\text{eps}} = 10$. Fig.~\ref{fig:Q_safe} shows the change of Q function with and without the safety filter. We can see that the Q function will converge in 1500 iterations, and taking action $u=1$ will give the highest Q function value as expected.
With the safety filter, the Q function converges to lower values, as the safety filter will limit the visitation of those high-rewarding but unsafe regions in the state space.
If we consider greedy policy based on the learned Q function (take action that gives maximized Q function, i.e., $\pi^Q(x,u) = \arg\max_u Q(x, u)$), we should expect similar control behaviors as policy gradient with and without the safety filter.
\begin{figure}
    \centering
    \includegraphics[width=1\linewidth]{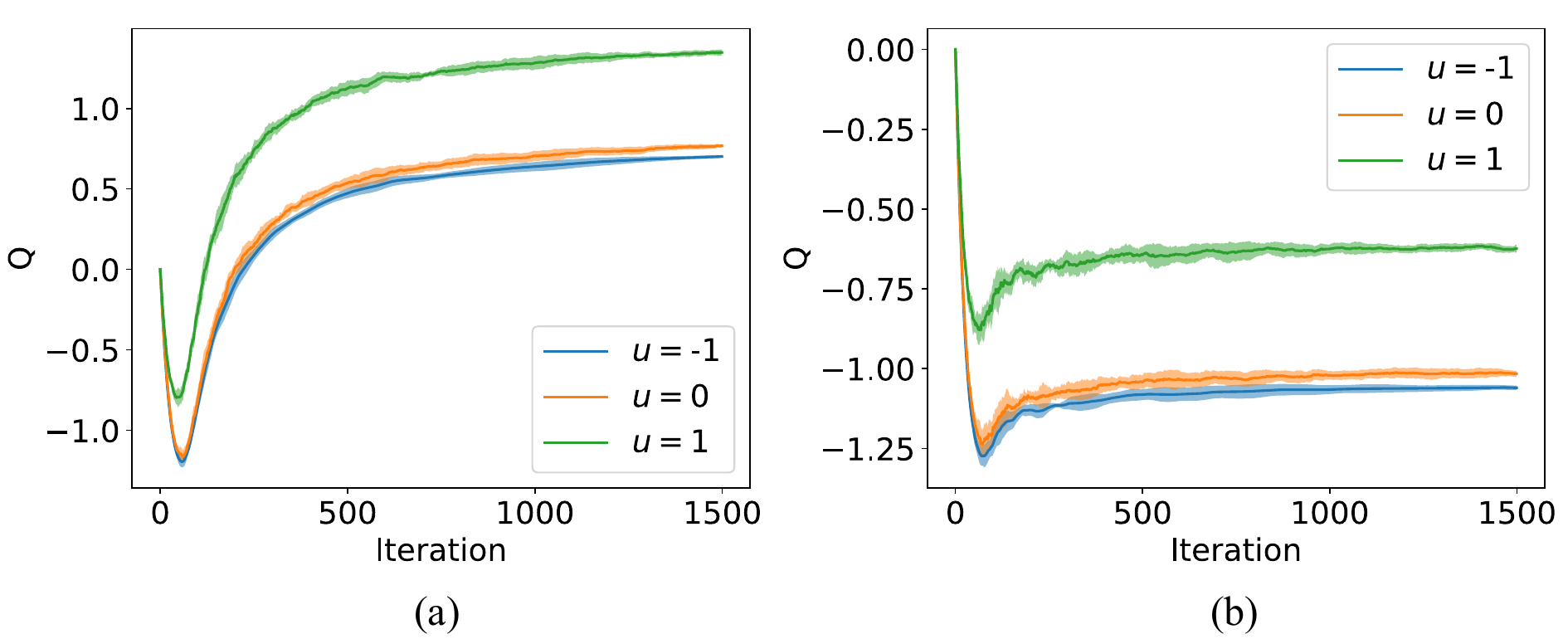}
    \caption{Change of Q function at initial state $x_0=0$ (a) without with the safety filter; (b) with the safety filter.}
    \label{fig:Q_safe}
\end{figure}

%%%%%%%%%%%%%%%%%%%%%%%%%%%%%%%%%%%%%%%%%%%%%%%%%%%%%%%%%%%%%%%%%%%%%%%%%%%%%
\section{Conclusion}
\label{sec:conclusion}
This paper focuses on the problem of ensuring long-term safety with high probability in stochastic systems. 
The major challenge of this problem is the stringent tradeoffs between longer-term safety vs. computational burdens. To mitigate the tradeoffs, we explore how to impose forward invariance on a probability space. Even when set invariance on the state space is satisfied with high probability at each time, long-term safety may not be guaranteed due to the accumulation of uncertainty and risk over time. In contrast, imposing probabilistic invariance on long-term probability allows myopic conditions/controllers to assure long-term design specifications. 
We then integrate this technique into both control and learning methods. The advantages of the proposed control and learning methods are demonstrated using numerical examples. 
Beyond these contributions, probabilistic invariance naturally extends to a broad range of scenarios, including systems with unknown dynamics~\cite{gangadhar2022adaptive}, multi-agent coordination~\cite{jing2022probabilistic}, and latent risks with unobservable state~\cite{gangadhar2023occlusion}. 

Future research directions include generalizing the proposed approach to systems that are not affine control and to new control techniques.  
The requirement for affine control systems is common in many safe control techniques for deterministic systems. For stochastic systems, even for affine control systems, the safety conditions (e.g., the constraints for control barrier function conditions to hold with high probability) may not give rise to linear action constraints. In the proposed approach, the linearity of the safety constraint \eqref{eq:infgen_to_mf} requires affine control dynamics, while Theorem 1 generalizes to non-affine systems. This suggests the relaxation for affine control reduces to developing efficient techniques for evaluating the safety constraints \eqref{eq:infgen_to_mf} for non-affine-control systems. Additionally, the probabilistic invariance concept can be leveraged to tackle other control problems traditionally approached via set invariance, such as constrained optimal control, stability analysis, and robust control~\cite{blanchini1999set}. Since probabilistic invariance shares conceptual parallels with set invariance, it may provide new insights and alternative methodologies for such problems as well.

%Besides, efficient generalization to high-dimensional systems is an exciting future direction, as the computation of the safety probability can be challenging with the increasing system dimension. Certain dimensionality reduction techniques can be of help, for example finding low-dimensional features of high-dimensional systems using comparison lemma~\cite{ikeda1977comparison}.

%We expect such algorithm can be usefully for methodologies such as Lyapunov methods~\cite{hahn1967stability,lin1996smooth} and receding horizon control~\cite{propoi1963use,keerthi1988optimal} that are previously studied only with set invariance on the state space.

\bibliographystyle{ieeetr}
\bibliography{citation}

\appendix

\subsection{Safe Control Algorithms for Comparison}
\label{sec:algs_for_comparison}
We compare our proposed controller with three existing safe controllers designed for stochastic systems. Below, we present their techniques in the settings of this paper. We consider long-term safety in~\eqref{eq:long_term_safety} with fixed time horizon and time-invariant zero margin, \ie $\mP \left(\minf_{X_t}(H) \geq 0 \right)$.
\begin{itemize}[leftmargin=*]
    \item Proposed controller: The safety condition is given by
    \begin{align}
        D_\sprob(Z_t, U_t) \geq - \alpha (\sprob(Z_t) - (1-\epsilon)), 
    \end{align}
    where $\alpha > 0$ is set to be a constant, and $\sprob$ is defined based on type 1 in~\eqref{eq:cases_summary}.
    
    \item Stochastic control barrier functions (StoCBF) \cite{clark2019control}: The safety condition is given by 
    \begin{align}
        \label{eq:StoCBF safe constraint}
        % & \frac{\partial \phi}{\partial x} f(X_t) + \frac{\partial \phi}{\partial x} g(X_t) U_t + \frac{1}{2}
        % \textbf{tr}\left(\sigma(X_t)^{T}\frac{\partial^2 \phi}{\partial x^2}\sigma(X_t)\right) \\
        % & \leq \alpha(\phi(X_t)).
        D_\phi(X_t, U_t) \geq - \eta \phi(X_t),
        % D_\phi(X_t, U_t) =\frac{d\phi}{dx}E[\frac{dx}{dt}]\\
        % \phi(x_{t+1}) \geq -\alpha \phi(x_t)
        % \\
        %  E[\phi(x_{t+1})]   \geq \phi(x_{t}) -\alpha h \phi(x_t) =  (1-\alpha h) \phi(x_t)  \\
        %  h: sampling interval
    \end{align}
    where $\eta > 0$ is a constant.
    Here, the mapping $D_\phi:\R^{n}\times\R^m \rightarrow \R$ is defined as the infinitesimal generator of the stochastic process $X_t$ acting on the barrier function $\phi$, \ie
    \begin{align}
    \label{eq:infgen_to_phi}
        \begin{split}
            D_\phi (X_t,U_t) := &A\phi (X_t) \\
            = &\gL_{f}\phi(X_t)+\gL_{g}\phi(X_t) U_t \\ 
            &+\frac{1}{2} \text{tr} \left(\left[\sigma(X_t)\right]\left[\sigma(X_t)\right]^\top\Hess \phi(X_t)\right).
        \end{split}
    \end{align} 
    
    This condition constrains the average system state to move within the tangent cone of the safe set. 
    % The hyper-parameter used in the simulation is $\alpha = 9$.
    \item Probabilistic safety barrier certificates (PrSBC) \cite{luo2019multi}: The safety condition is given by
    \begin{align}
        \label{eq:PrSBC safe constraint}
        \mathbb{P}\left(D_\phi(X_t, U_t) + \eta \phi(X_t) \geq 0\right) \geq 1 - \epsilon,
    \end{align}
    where  $\eta > 0$ is a constant. This condition constrains the state to stay within the safe set in the infinitesimal future interval with high probability. 
    % The hyper-parameters used in the simulation are $\gamma = 0.1$, $\epsilon = 0.9$. 
    \item Conditional-value-at-risk barrier functions (CVaR) \cite{ahmadi2020risk}: The safety condition is given by
    \begin{align}
        \label{eq:CVaR safe constraint}
        \text{CVaR}_\epsilon \left(\phi(X_{t_{k+1}})\right) \geq \gamma \phi(X_{t_k})
    \end{align}
    where $\gamma \in (0,1)$ is a constant, $\{t_0 = 0, t_1, t_2,\cdots\}$ is a discrete sampled time of equal sampling intervals. 
    This is a sufficient condition to ensure the value of $\text{CVaR}^{k}_\epsilon(\phi(X_{t_k}))$ conditioned on $X_0 = x$ to be non-negative at all sampled time $t_{k \in \mathbb{Z}_+}$. The value of $\text{CVaR}^{k}_\epsilon(\phi(X_{t_k}))$ quantifies the evaluation made at time $t_0 = 0$ about the safety at time $t_k$.
    % The hyper-parameters used in the simulation are $\alpha = 0.7$, $\beta = 0.1$. 
\end{itemize}

The parameters used are listed in Table \ref{tb:parameter list}. Note that we use the same values for similar parameters in different methods for fair comparisons. 
Detailed discussion on the CVaR parameters can be found in Appendix~\ref{sec:cvar_comparison}.

\begin{table}[h]
\caption{Parameters used in simulation.}
% \vspace{-1mm}
\label{tb:parameter list}
\begin{center}
\begin{tabular}{ |c|c| } 
\hline
\textbf{Controller} & \textbf{Parameters}  \\
\hline
\hline
Proposed controller & $\alpha(x) = \eta x$, $\eta = 1$ $\epsilon = 0.1$, $H=10$ \\
\hline
StoCBF & $\eta = 1$ \\
\hline
PrSBC & $\eta = 1$, $\epsilon = 0.1$ \\
\hline
CVaR & $\gamma = 0.65$, $\epsilon = 0.1$ \\
\hline
\end{tabular}
\end{center}
% \vspace{-1mm}
\end{table}

% \begin{comment}

\subsection{Detailed Comparison with CVaR Baseline}
\label{sec:cvar_comparison}

In this section, we provide a detailed comparison of the CVaR~\cite{ahmadi2020risk} baseline and the proposed method for safe control. Fig.~\ref{fig:cvar_ablation} and Fig.~\ref{fig:proposed_ablation} show the behaviors of CVaR and the proposed method for varying parameters. 
It is worth noting that the guarantee provided by CVaR is for each time and is different from the guarantee of the proposed method. Due to this difference, a valid CVaR can still have a low long-term risk probability.

\begin{figure}
    \centering
    \includegraphics[width=1\linewidth]{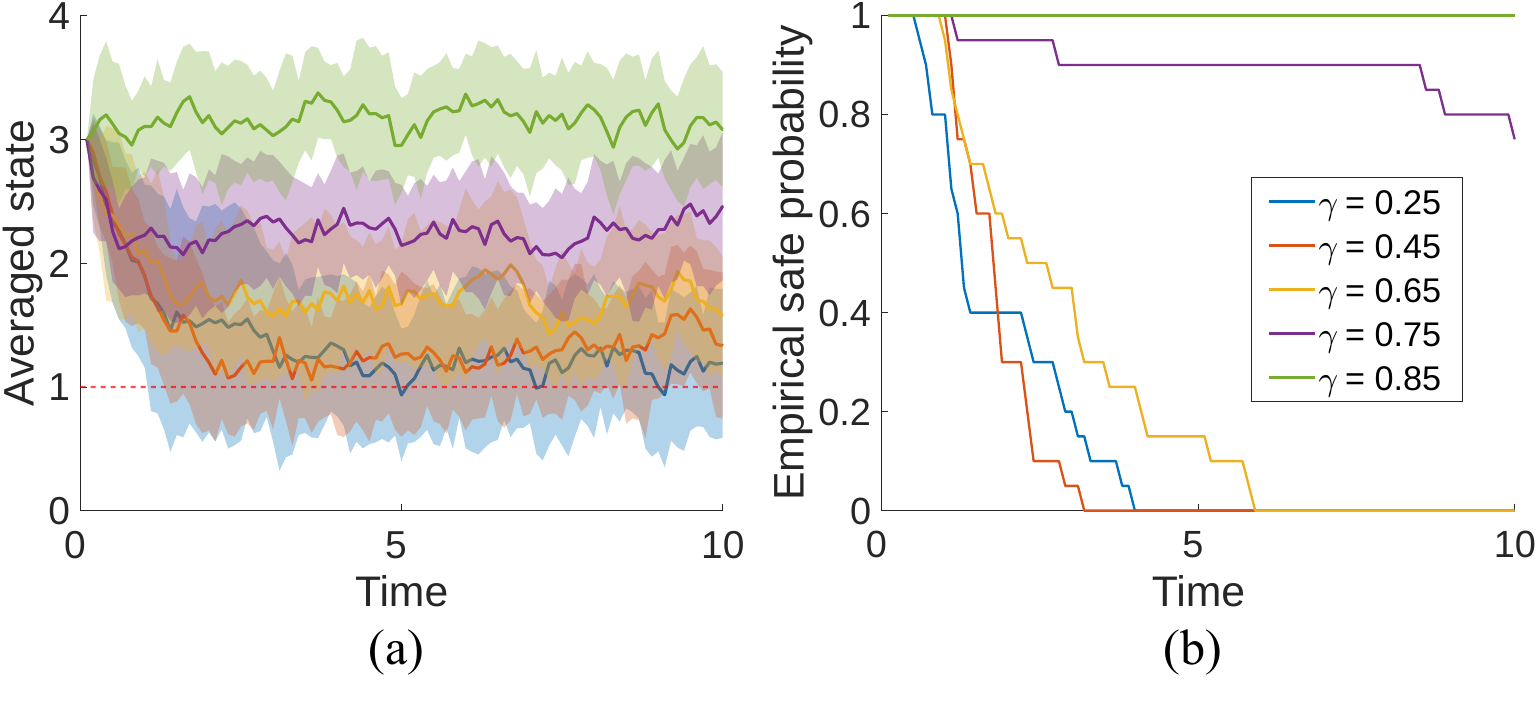}
    \caption{CVaR performance with different choices of $\gamma$. 
    }
    \label{fig:cvar_ablation}
\end{figure}

\begin{figure}
    \centering
    \includegraphics[width=1\linewidth]{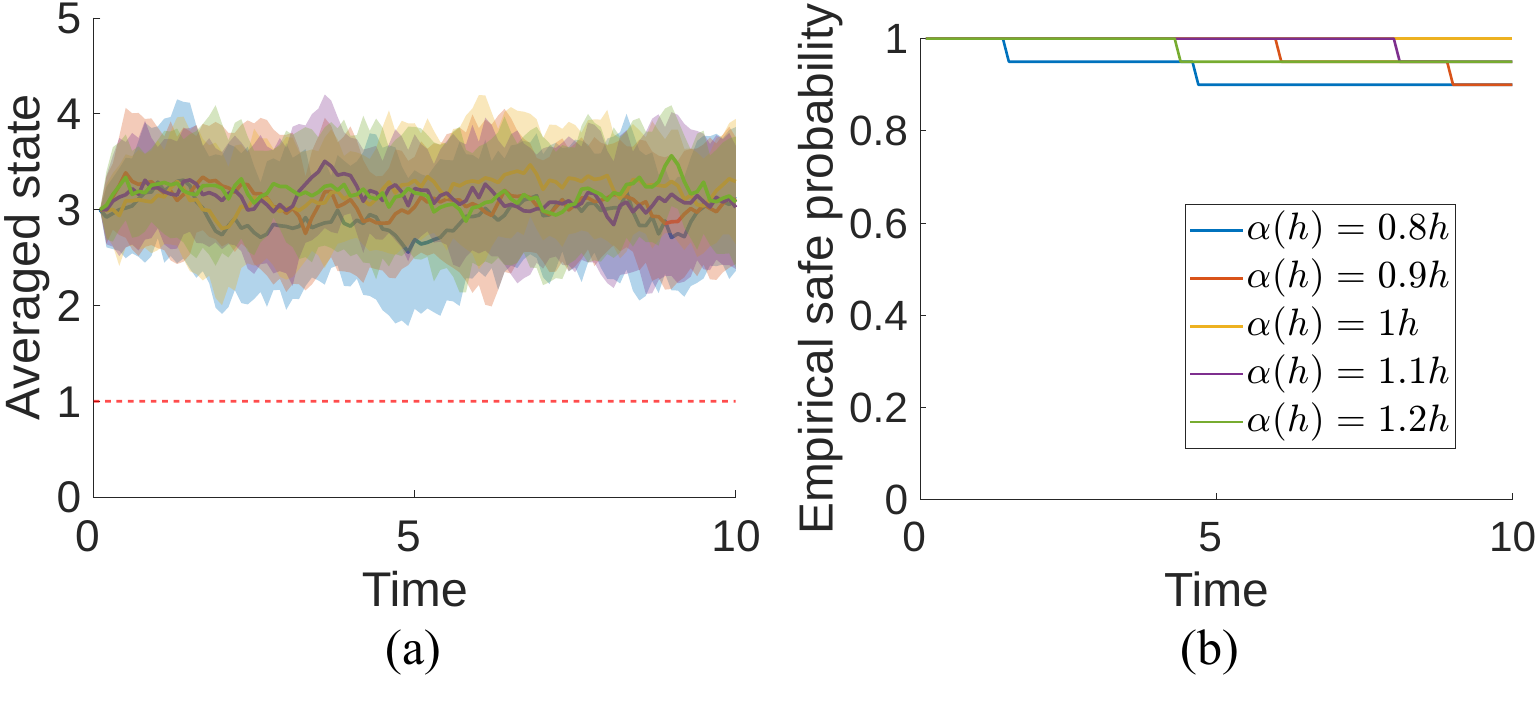}
    \caption{Performance of the proposed method with different choices of $\alpha$. 
    }
    \label{fig:proposed_ablation}
\end{figure}

\subsection{Neural Network-based Nominal Controller}
\label{sec:nn_nominal_exp}
In this section, we show the efficacy of the proposed method for safe control with neural network-based nominal controllers.

We consider the linear control affine system~\eqref{eq:x_trajectory} with $f(X_t) \equiv A = 2$, $g(X_t) \equiv 1$, $\sigma(X_t) \equiv 2$. 
The safe set is defined as in~\eqref{eq:safe_set_experiment}. We use a 2-layer neural network $h^2$ as the nominal controller, with $W_1 = 0.5$, $W_2 = 1$, $b_1 = 0.2$, $b_2 = -0.5$. 
% \todo{fill in the values} 
The activation function $\sigma_{\text{act}}$ is chosen to be the ReLU function. The nominal controller has the form
\begin{equation}
\label{eq:NN_nominal}
\begin{aligned}
    N(X) & = \sigma_{\text{act}}\left(W_2 \cdot \sigma_{\text{act}} (W_1 X + b_1) + b_2\right).
\end{aligned}
\end{equation}
In the safety probability estimation and safe control phase, we only have access to the value of $N(X)$ instead of the exact expression~\eqref{eq:NN_nominal}. We show the results for the proposed safe control method in Fig.~\ref{fig:NN_nominal}. Since we can not write down the closed-loop dynamics of the system given the black-boxed neural network controller, none of the existing methods being compared in the previous section can be used. We can see that the nominal NN controller will yield unsafe behaviours, while with our proposed strategy long-term safety is ensured.

\begin{figure}
\centering
\includegraphics[width=1\linewidth]{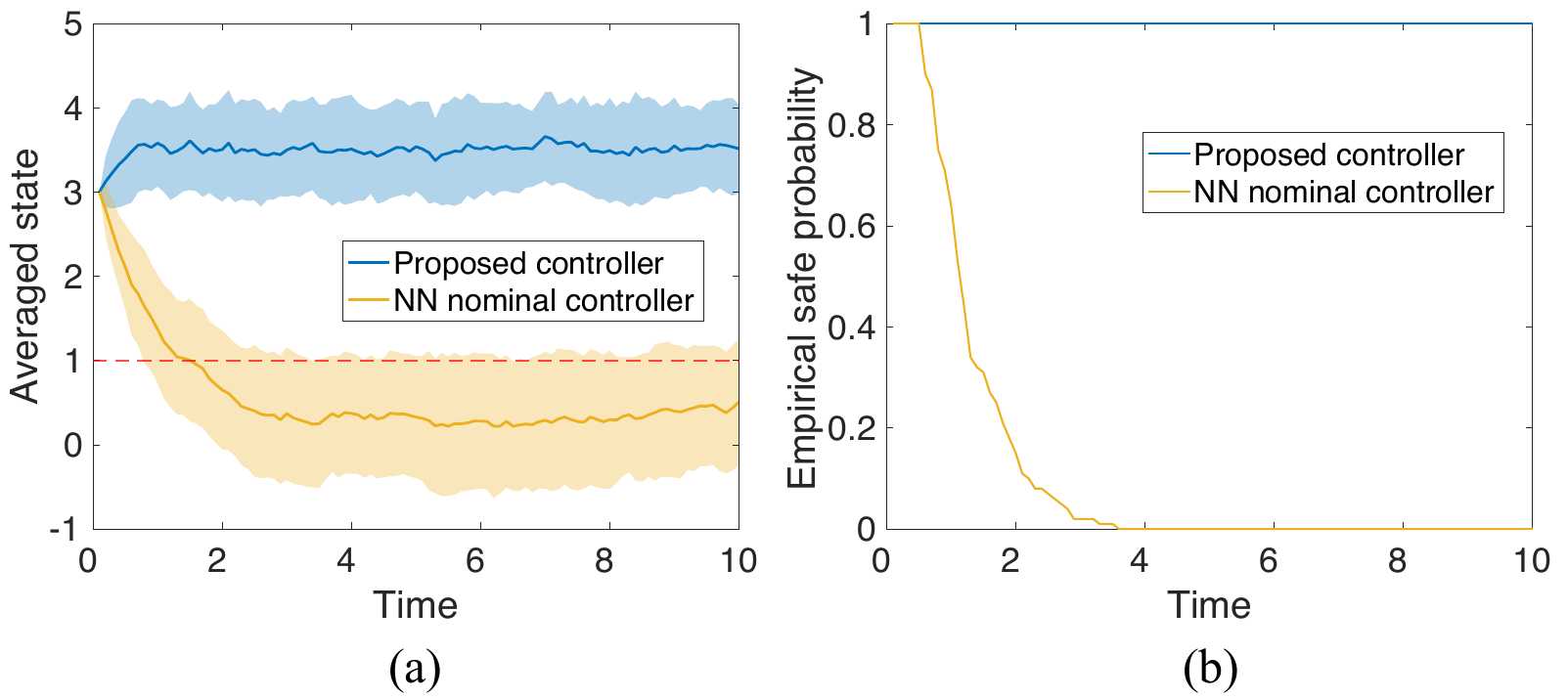}
\caption{Results in the switching control setting with NN nominal controllers. (a) the average system state over 100 trajectories with its standard deviation. Red dotted line indicates the boundary of the safe set. (b) the empirical safe probability. 
% \todo{change font size in plots}
\label{fig:NN_nominal}}
\end{figure}

\subsection{Results on 3-Dimensional Systems}
\label{sec:3d_results}
In this section, we show efficacy of the proposed safe control on a 3-dimensional system. We consider dynamics
\begin{equation}
\label{eq:x_trajectory_3d}
\begin{aligned}
    & \left[\frac{dx_1}{dt}, \frac{dx_2}{dt}, \frac{dx_3}{dt}\right] \\
    = \; & [A_1 x_1, A_2 x_2, A_3 x_3] + [u, u, u] + \sigma \; dW_t / dt,
\end{aligned}
\end{equation}
where $[x_1, x_2, x_3] := X \in \mathbb{R}^3$ is the state, $u \in \mathbb{R}$ is the control, and $W_t$ is a 3-dimensional Wiener process with $W_0 = \mathbf{0}$. Here, $[A_1, A_2, A_3] = [-0.5, -0.6, -0.7]$ and we assume $\sigma = 1$ for simplicity. 
The nominal control is chosen to be $u_{\text{nominal}} = 0$.
The safe set is defined as $\mathcal{C} = \{X: \min(x_1, x_2, x_3) \geq 1\}$. We run our proposed safe control method with risk tolerance $\epsilon = 0.1$ for $T = 10 \mathrm{s}$ with $dt = 0.05$. The initial state is set to be $X_0 = [5, 5, 5]$.
The results are shown in Fig.~\ref{fig:3d_results}, where we can see on the left that the proposed safe control will maintain the system state within the safe enough region, and on the right that the average system state yields high expected long-term safety probability. The empirical safety probability is 1 at all time steps.

\begin{figure}
\centering
\includegraphics[width=1\linewidth]{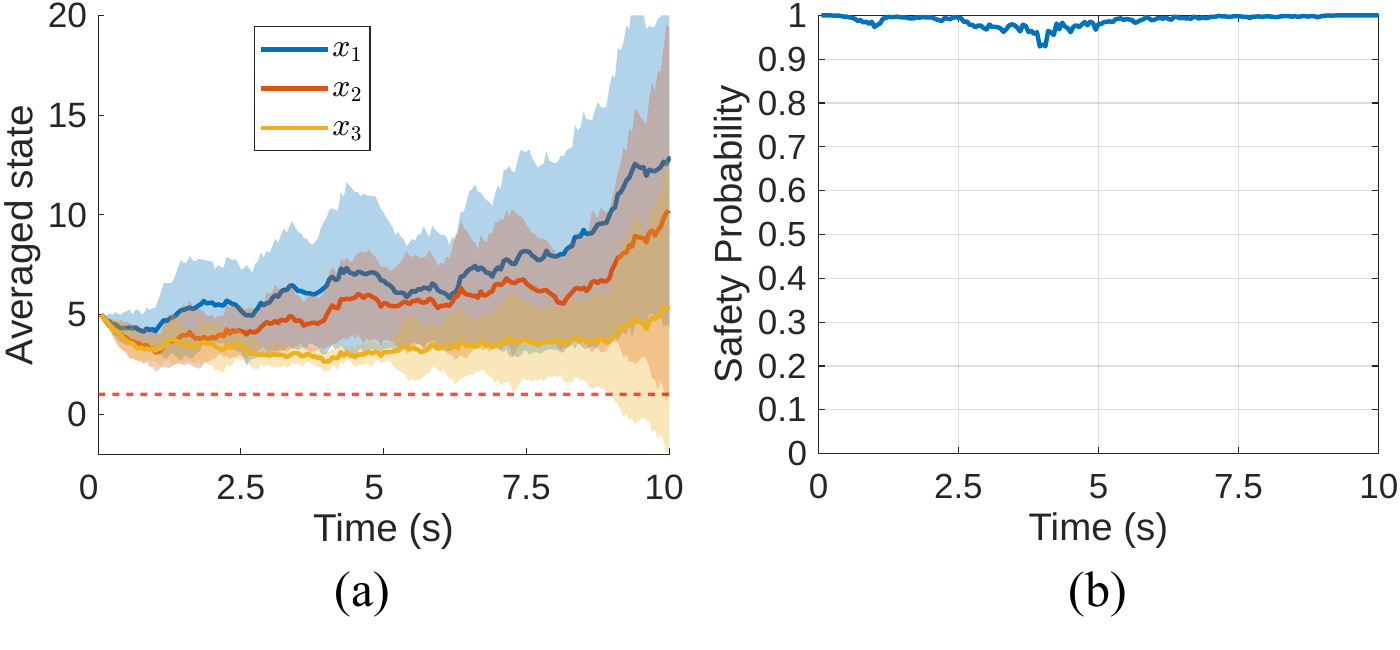}
\caption{Results on 3-dimensional systems. (a) the average system state over 10 trajectories with its standard deviation. Red dotted line indicates the boundary of the safe set. (b) the expected long-term safe probability. 
\label{fig:3d_results}}
\end{figure}

% \end{comment}

\begin{IEEEbiography}[{\includegraphics[width=1in,height=1.25in,clip,keepaspectratio]{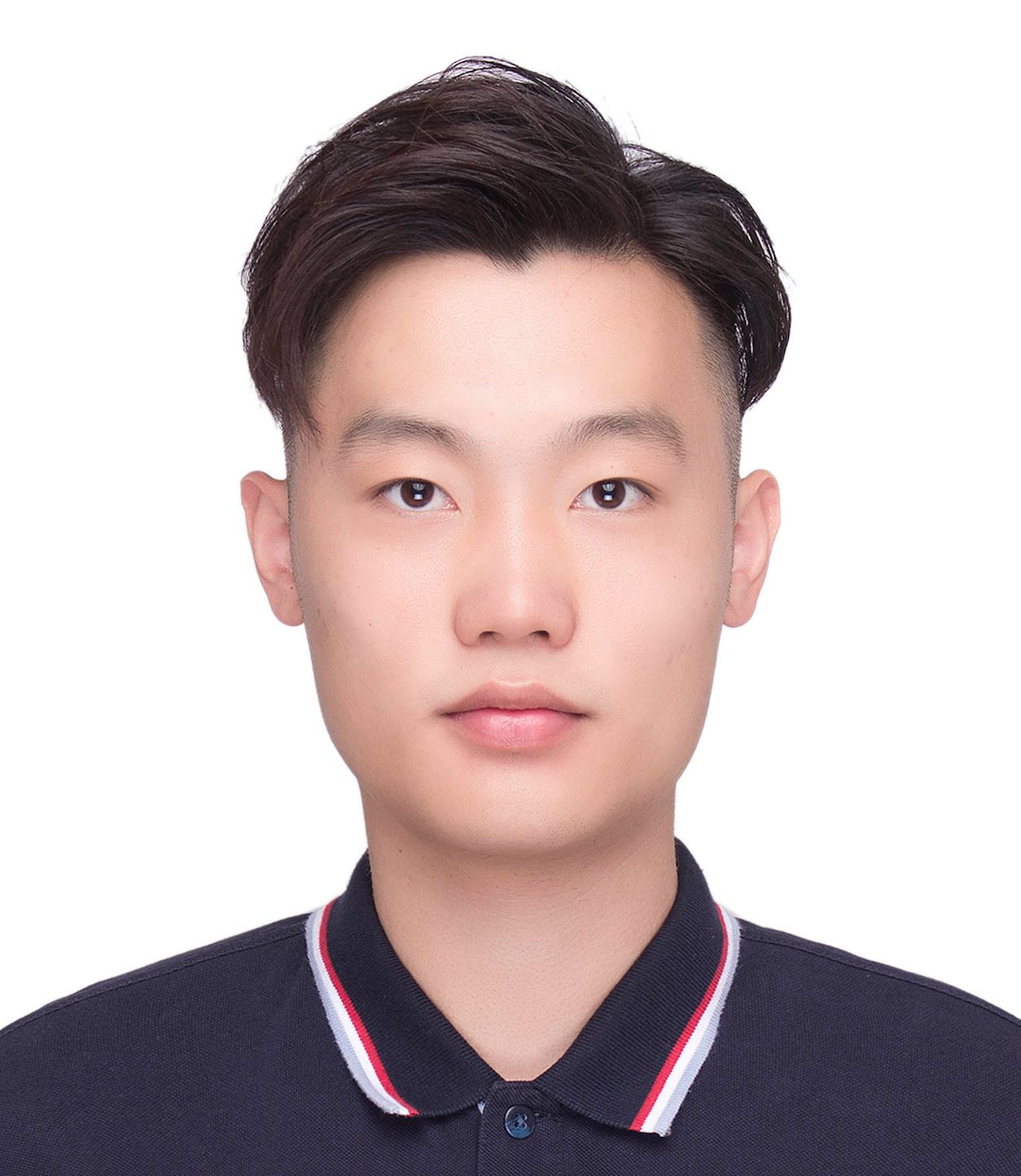}}]{Zhuoyuan Wang} received his B.E. degree in Automation from Tsinghua University, Beijing, China, in 2020 and is currently pursuing a Ph.D. degree in 
Electrical and Computer Engineering at Carnegie Mellon University, Pittsburgh, PA, USA.

His research interests include safety-critical control for stochastic systems, physics-informed learning, safe reinforcement learning and application to robotic systems.
He is a recipient of the Michel and Kathy Doreau Graduate Fellowship at Carnegie Mellon University.
\end{IEEEbiography}

\begin{IEEEbiography}[{\includegraphics[width=1in,height=1.25in,clip,keepaspectratio]{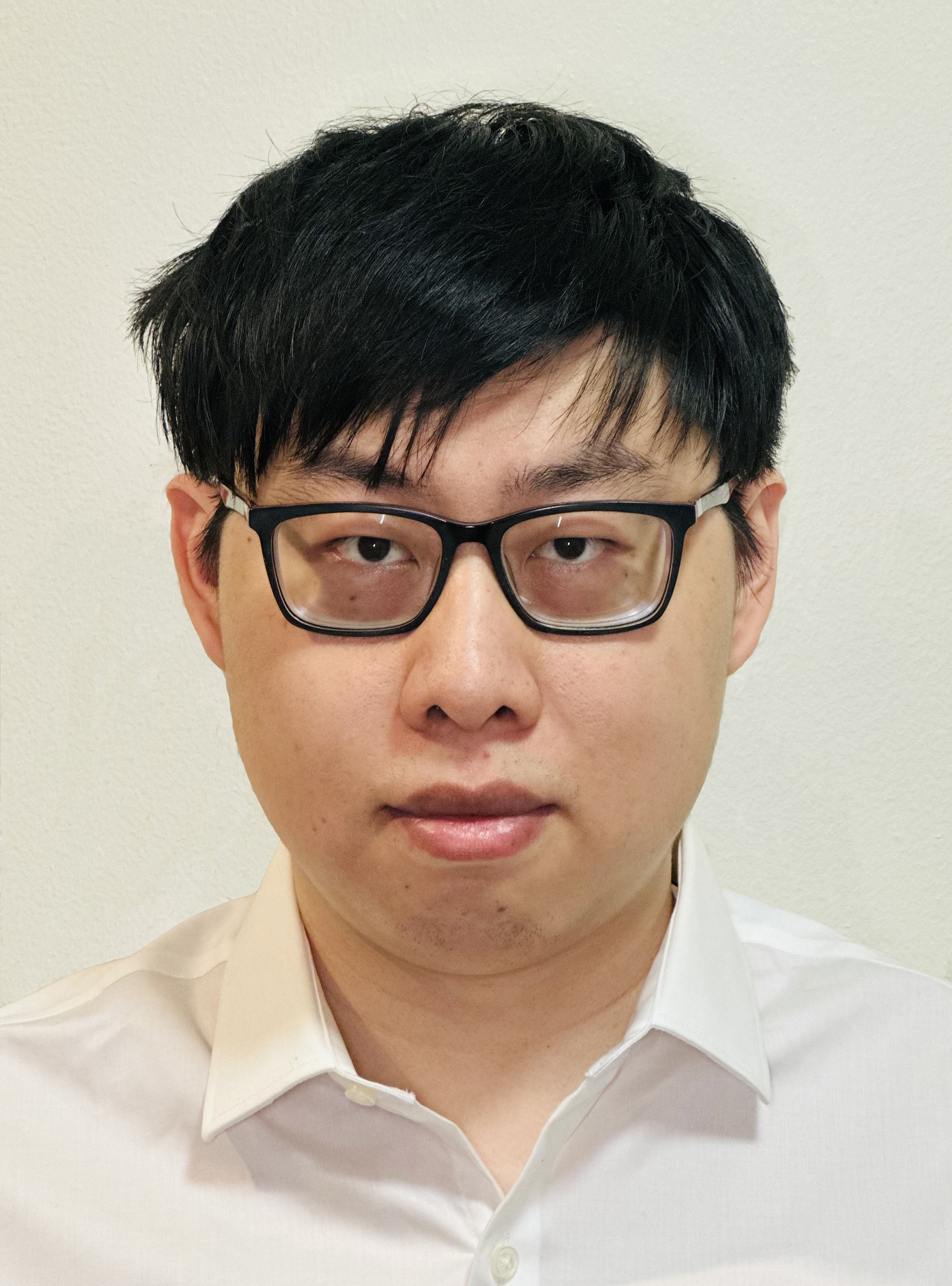}}]{Haoming Jing} received his B.S. degree in Electrical Engineering from University of California, Santa Barbara, Santa Barbara, CA, USA in 2020 and is currently pursuing a Ph.D. degree in 
Electrical and Computer Engineering at Carnegie Mellon University, Pittsburgh, PA, USA.

His research interests include safety-critical control for multi-agent stochastic systems, safety for autonomous driving systems, and safe human-machine interactions.
\end{IEEEbiography}

\begin{IEEEbiography}[{\includegraphics[width=1in,clip,keepaspectratio]{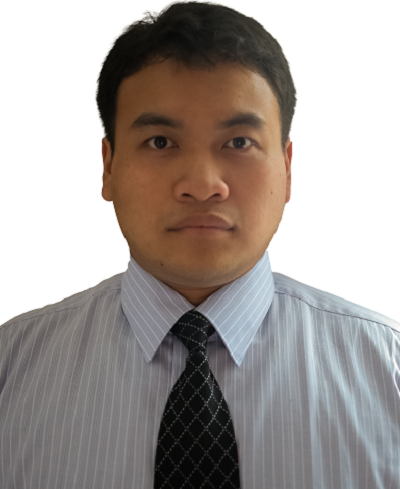}}]{Christian Kurniawan}
received a B.Sc. degree in Mathematics from Brigham Young University, Hawaii campus in 2014 and an M.S. degree in Mechanical Engineering with an emphasis on computational grain boundary from Brigham Young University's main campus in 2018. His current research interest includes applications of optimization techniques and inverse problem theory in various computational engineering problems and autonomous systems.
\end{IEEEbiography}

\begin{IEEEbiography}[{\includegraphics[width=1in,height=1.25in,clip,keepaspectratio]{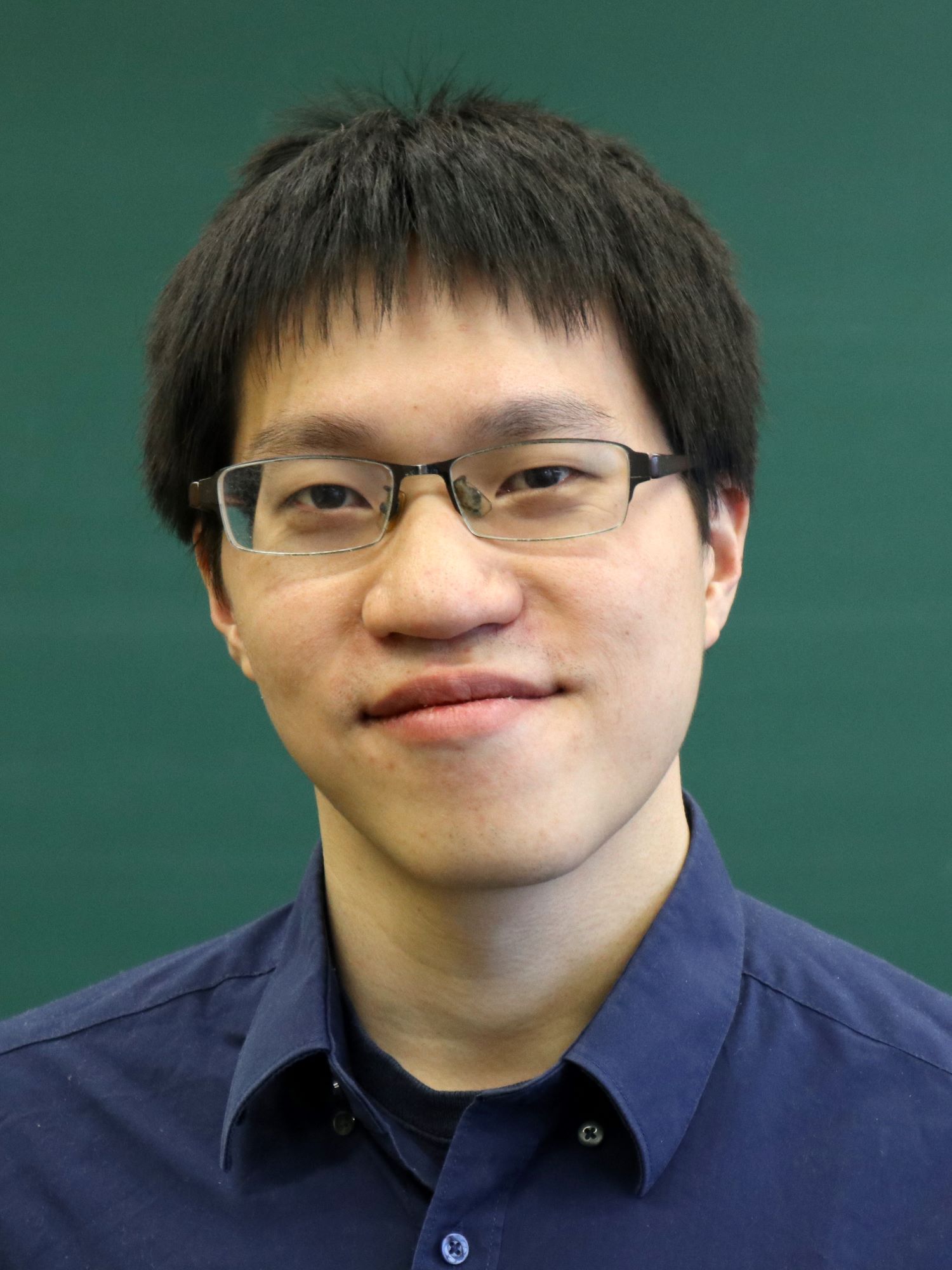}}]{Albert Chern} is an Assistant Professor in Computer Science and Engineering at University of California San Diego, La Jolla, CA, USA, since 2020.  He received his Ph.D. in Applied and Computational Mathematics at California Institute of Technology, Pasadena, CA, USA, in 2017, and was a Postdoctoral Researcher in Mathematics at Technische Universit{\"a}t Berlin, Berlin, Germany, from 2017 to 2020.  Chern's research interest lies in applications of differential geometry to computational math, fluid dynamics, computer graphics, and the interplay between stochastic processes and differential equations.  Chern was a recipient of the NSF CAREER Award in 2023.
\end{IEEEbiography}

\begin{IEEEbiography}[{\includegraphics[width=1in,height=1.25in,clip,keepaspectratio]{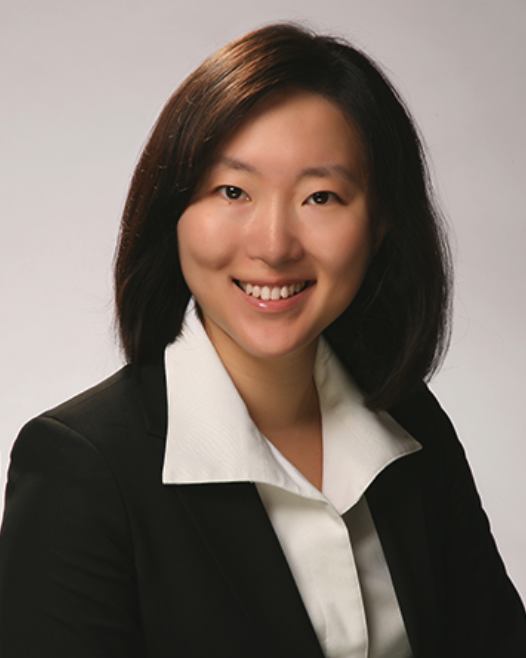}}]{Yorie Nakahira} is an Assistant Professor in the Department of Electrical and Computer Engineering at Carnegie Mellon University. She received B.E. in Control and Systems Engineering from Tokyo Institute of Technology in 2012 and Ph.D. in Control and Dynamical Systems from California Institute of Technology in 2019. Her research interests include the fundamental theory of optimization, control, and learning and its application to neuroscience, cell biology, smart grid, cloud computing, finance, autonomous robots.
\end{IEEEbiography}

\end{document}

%% file: MACPdef.tex
\newcommand{\mE}{\ensuremath{\mathbb{E}}}
\newcommand{\mF}{\ensuremath{\mathbf{F}}}

\newcommand{\mP}{\ensuremath{\mathbb{P}}}

\newcommand\gC{{\mathcal{C}}}

\newcommand\gE{{\mathcal{E}}}

\newcommand\gL{{\mathcal{L}}}

\newcommand\gX{{\mathcal{X}}}
\newcommand\gY{{\mathcal{Y}}}

\providecommand{\keywords}[1]{\vspace*{0.5\baselineskip}\par\noindent\textbf{Keywords:} #1}

{%
\begin{list}{#1}{
\vspace{-\topsep}
\vspace{-\partopsep}
\setlength{\itemindent}{0cm}
\setlength{\rightmargin}{0cm}
\setlength{\listparindent}{0cm}
\settowidth{\labelwidth}{#1}
\setlength{\leftmargin}{\labelwidth}
\addtolength{\leftmargin}{\labelsep}
\setlength{\itemsep}{0cm}
}%
}%
{%
\end{list}
\vspace{-\topsep}
\vspace{-\partopsep}
}

{\begin{enumerate}%
\renewcommand{\theenumi}{\roman{enumi}}%
}%
{\end{enumerate}}

\usepackage{amsthm}

\theoremstyle{plain}
\newtheorem{theorem}{Theorem}

\newtheorem{lemma}{Lemma}

\theoremstyle{definition}
\newtheorem{definition}{Definition}

\newtheorem{example}{Example}
\theoremstyle{plain}% default

\theoremstyle{remark}
\newtheorem{remark}{Remark}

%% file: arxiv.bbl
\begin{thebibliography}{10}

\bibitem{blanchini1999set}
F.~Blanchini, ``Set invariance in control,'' {\em Automatica}, vol.~35, no.~11, pp.~1747--1767, 1999.

\bibitem{khalil2002nonlinear}
H.~Khalil, {\em Nonlinear Systems}.
\newblock Pearson Education, Prentice Hall, 2002.

\bibitem{ames2019control}
A.~D. Ames, S.~Coogan, M.~Egerstedt, G.~Notomista, K.~Sreenath, and P.~Tabuada, ``Control barrier functions: Theory and applications,'' in {\em 2019 18th European Control Conference (ECC)}, pp.~3420--3431, IEEE, 2019.

\bibitem{nagumo1942lage}
M.~Nagumo, ``{\"U}ber die lage der integralkurven gew{\"o}hnlicher differentialgleichungen,'' {\em Proceedings of the Physico-Mathematical Society of Japan. 3rd Series}, vol.~24, pp.~551--559, 1942.

\bibitem{bony1969principe}
J.-M. Bony, ``Principe du maximum, in{\'e}galit{\'e} de harnack et unicit{\'e} du probleme de cauchy pour les op{\'e}rateurs elliptiques d{\'e}g{\'e}n{\'e}r{\'e}s,'' in {\em Annales de l'institut Fourier}, vol.~19, pp.~277--304, 1969.

\bibitem{brezis1970characterization}
H.~Brezis, ``On a characterization of flow-invariant sets,'' {\em Communications on Pure and Applied Mathematics}, vol.~23, no.~2, pp.~261--263, 1970.

\bibitem{clark2021control}
A.~Clark, ``Control barrier functions for stochastic systems,'' {\em Automatica}, vol.~130, p.~109688, 2021.

\bibitem{9561894}
Y.~Lyu, W.~Luo, and J.~M. Dolan, ``Probabilistic safety-assured adaptive merging control for autonomous vehicles,'' in {\em 2021 IEEE International Conference on Robotics and Automation (ICRA)}, pp.~10764--10770, 2021.

\bibitem{so2025comment}
O.~So and C.~Fan, ``Comment (s) on “control barrier functions for stochastic systems”[automatica 130 (2021) 109688],'' {\em Automatica}, p.~112185, 2025.

\bibitem{luo2019multi}
W.~Luo, W.~Sun, and A.~Kapoor, ``Multi-robot collision avoidance under uncertainty with probabilistic safety barrier certificates,'' {\em arXiv preprint arXiv:1912.09957}, 2019.

\bibitem{abate2008probabilistic}
A.~Abate, M.~Prandini, J.~Lygeros, and S.~Sastry, ``Probabilistic reachability and safety for controlled discrete time stochastic hybrid systems,'' {\em Automatica}, vol.~44, no.~11, pp.~2724--2734, 2008.

\bibitem{chapman2019risk}
M.~P. Chapman, J.~Lacotte, A.~Tamar, D.~Lee, K.~M. Smith, V.~Cheng, J.~F. Fisac, S.~Jha, M.~Pavone, and C.~J. Tomlin, ``A risk-sensitive finite-time reachability approach for safety of stochastic dynamic systems,'' in {\em 2019 American Control Conference (ACC)}, pp.~2958--2963, IEEE, 2019.

\bibitem{kariotoglou2013approximate}
N.~Kariotoglou, S.~Summers, T.~Summers, M.~Kamgarpour, and J.~Lygeros, ``Approximate dynamic programming for stochastic reachability,'' in {\em 2013 European Control Conference (ECC)}, pp.~584--589, IEEE, 2013.

\bibitem{abate2006probabilistic}
A.~Abate, S.~Amin, M.~Prandini, J.~Lygeros, and S.~Sastry, ``Probabilistic reachability and safe sets computation for discrete time stochastic hybrid systems,'' in {\em Proceedings of the 45th IEEE Conference on Decision and Control}, pp.~258--263, IEEE, 2006.

\bibitem{liao2022probabilistic}
W.~Liao, T.~Liang, X.~Wei, and Q.~Yin, ``Probabilistic reach-avoid problems in nondeterministic systems with time-varying targets and obstacles,'' {\em Applied Mathematics and Computation}, vol.~425, p.~127054, 2022.

\bibitem{vasileva2020probabilistic}
M.~Vasileva, F.~Shmarov, and P.~Zuliani, ``Probabilistic reachability for uncertain stochastic hybrid systems via gaussian processes,'' in {\em 2020 18th ACM-IEEE International Conference on Formal Methods and Models for System Design (MEMOCODE)}, pp.~1--11, IEEE, 2020.

\bibitem{bansal2017hamilton}
S.~Bansal, M.~Chen, S.~Herbert, and C.~J. Tomlin, ``Hamilton-jacobi reachability: A brief overview and recent advances,'' in {\em 2017 IEEE 56th Annual Conference on Decision and Control (CDC)}, pp.~2242--2253, IEEE, 2017.

\bibitem{choi2021robust}
J.~J. Choi, D.~Lee, K.~Sreenath, C.~J. Tomlin, and S.~L. Herbert, ``Robust control barrier-value functions for safety-critical control,'' in {\em 2021 60th IEEE Conference on Decision and Control (CDC)}, pp.~6814--6821, IEEE, 2021.

\bibitem{vzikelic2023learning}
{\DJ}.~{\v{Z}}ikeli{\'c}, M.~Lechner, T.~A. Henzinger, and K.~Chatterjee, ``Learning control policies for stochastic systems with reach-avoid guarantees,'' in {\em Proceedings of the AAAI Conference on Artificial Intelligence}, vol.~37, pp.~11926--11935, 2023.

\bibitem{clark2019control}
A.~Clark, ``Control barrier functions for complete and incomplete information stochastic systems,'' in {\em 2019 American Control Conference (ACC)}, pp.~2928--2935, IEEE, 2019.

\bibitem{wang2021safety}
C.~Wang, Y.~Meng, S.~L. Smith, and J.~Liu, ``Safety-critical control of stochastic systems using stochastic control barrier functions,'' in {\em 2021 60th IEEE Conference on Decision and Control (CDC)}, pp.~5924--5931, IEEE, 2021.

\bibitem{santoyo2021barrier}
C.~Santoyo, M.~Dutreix, and S.~Coogan, ``A barrier function approach to finite-time stochastic system verification and control,'' {\em Automatica}, vol.~125, p.~109439, 2021.

\bibitem{nishimura2024control}
Y.~Nishimura and K.~Hoshino, ``Control barrier functions for stochastic systems and safety-critical control designs,'' {\em IEEE Transactions on Automatic Control}, 2024.

\bibitem{singletary2022safe}
A.~Singletary, M.~Ahmadi, and A.~D. Ames, ``Safe control for nonlinear systems with stochastic uncertainty via risk control barrier functions,'' {\em IEEE Control Systems Letters}, vol.~7, pp.~349--354, 2022.

\bibitem{samuelson2018safety}
S.~Samuelson and I.~Yang, ``Safety-aware optimal control of stochastic systems using conditional value-at-risk,'' in {\em 2018 Annual American Control Conference (ACC)}, pp.~6285--6290, IEEE, 2018.

\bibitem{prajna2007framework}
S.~Prajna, A.~Jadbabaie, and G.~J. Pappas, ``A framework for worst-case and stochastic safety verification using barrier certificates,'' {\em IEEE Transactions on Automatic Control}, vol.~52, no.~8, pp.~1415--1428, 2007.

\bibitem{yaghoubi2020risk}
S.~Yaghoubi, K.~Majd, G.~Fainekos, T.~Yamaguchi, D.~Prokhorov, and B.~Hoxha, ``Risk-bounded control using stochastic barrier functions,'' {\em IEEE Control Systems Letters}, vol.~5, no.~5, pp.~1831--1836, 2020.

\bibitem{huang2017probabilistic}
C.~Huang, X.~Chen, W.~Lin, Z.~Yang, and X.~Li, ``Probabilistic safety verification of stochastic hybrid systems using barrier certificates,'' {\em ACM Transactions on Embedded Computing Systems (TECS)}, vol.~16, no.~5s, pp.~1--19, 2017.

\bibitem{anand2019verification}
M.~Anand, P.~Jagtapt, and M.~Zamani, ``Verification of switched stochastic systems via barrier certificates,'' in {\em 2019 IEEE 58th Conference on Decision and Control (CDC)}, pp.~4373--4378, IEEE, 2019.

\bibitem{mazouz2022safety}
R.~Mazouz, K.~Muvvala, A.~Ratheesh~Babu, L.~Laurenti, and M.~Lahijanian, ``Safety guarantees for neural network dynamic systems via stochastic barrier functions,'' {\em Advances in Neural Information Processing Systems}, vol.~35, pp.~9672--9686, 2022.

\bibitem{gangadhar2022adaptive}
S.~Gangadhar, Z.~Wang, H.~Jing, and Y.~Nakahira, ``Adaptive safe control for driving in uncertain environments,'' in {\em 2022 IEEE Intelligent Vehicles Symposium (IV)}, pp.~1662--1668, IEEE, 2022.

\bibitem{Farina2016}
M.~Farina, L.~Giulioni, and R.~Scattolini, ``{Stochastic linear Model Predictive Control with chance constraints – A review},'' {\em Journal of Process Control}, vol.~44, pp.~53--67, aug 2016.

\bibitem{Hewing2020}
L.~Hewing, K.~P. Wabersich, M.~Menner, and M.~N. Zeilinger, ``{Learning-Based Model Predictive Control: Toward Safe Learning in Control},'' {\em Annual Review of Control, Robotics, and Autonomous Systems}, vol.~3, pp.~269--296, may 2020.

\bibitem{brudigam2021stochastic}
T.~Br{\"u}digam, M.~Olbrich, D.~Wollherr, and M.~Leibold, ``Stochastic model predictive control with a safety guarantee for automated driving,'' {\em IEEE Transactions on Intelligent Vehicles}, vol.~8, no.~1, pp.~22--36, 2021.

\bibitem{gu2022review}
S.~Gu, L.~Yang, Y.~Du, G.~Chen, F.~Walter, J.~Wang, Y.~Yang, and A.~Knoll, ``A review of safe reinforcement learning: Methods, theory and applications,'' {\em arXiv preprint arXiv:2205.10330}, 2022.

\bibitem{garcia2015comprehensive}
J.~Garc{\i}a and F.~Fern{\'a}ndez, ``A comprehensive survey on safe reinforcement learning,'' {\em Journal of Machine Learning Research}, vol.~16, no.~1, pp.~1437--1480, 2015.

\bibitem{xu2021crpo}
T.~Xu, Y.~Liang, and G.~Lan, ``Crpo: A new approach for safe reinforcement learning with convergence guarantee,'' in {\em International Conference on Machine Learning}, pp.~11480--11491, PMLR, 2021.

\bibitem{chen2021primal}
Y.~Chen, J.~Dong, and Z.~Wang, ``A primal-dual approach to constrained markov decision processes,'' {\em arXiv preprint arXiv:2101.10895}, 2021.

\bibitem{liu2021policy}
Y.~Liu, A.~Halev, and X.~Liu, ``Policy learning with constraints in model-free reinforcement learning: A survey,'' in {\em The 30th International Joint Conference on Artificial Intelligence (IJCAI)}, 2021.

\bibitem{wachi2020safe}
A.~Wachi and Y.~Sui, ``Safe reinforcement learning in constrained markov decision processes,'' in {\em International Conference on Machine Learning}, pp.~9797--9806, PMLR, 2020.

\bibitem{chow2018lyapunov}
Y.~Chow, O.~Nachum, E.~Duenez-Guzman, and M.~Ghavamzadeh, ``A lyapunov-based approach to safe reinforcement learning,'' {\em Advances in neural information processing systems}, vol.~31, 2018.

\bibitem{qin2022neural}
C.~Qin, J.~Wang, H.~Zhu, J.~Zhang, S.~Hu, and D.~Zhang, ``Neural network-based safe optimal robust control for affine nonlinear systems with unmatched disturbances,'' {\em Neurocomputing}, vol.~506, pp.~228--239, 2022.

\bibitem{wei2022safe}
T.~Wei and C.~Liu, ``Safe control with neural network dynamic models,'' in {\em Learning for Dynamics and Control Conference}, pp.~739--750, PMLR, 2022.

\bibitem{zhao2021learning}
H.~Zhao, X.~Zeng, T.~Chen, Z.~Liu, and J.~Woodcock, ``Learning safe neural network controllers with barrier certificates,'' {\em Formal Aspects of Computing}, vol.~33, pp.~437--455, 2021.

\bibitem{tasse2023rosarl}
G.~N. Tasse, T.~Love, M.~Nemecek, S.~James, and B.~Rosman, ``Rosarl: Reward-only safe reinforcement learning,'' {\em arXiv preprint arXiv:2306.00035}, 2023.

\bibitem{amani2021safe}
S.~Amani, C.~Thrampoulidis, and L.~Yang, ``Safe reinforcement learning with linear function approximation,'' in {\em International Conference on Machine Learning}, pp.~243--253, PMLR, 2021.

\bibitem{donti2020enforcing}
P.~L. Donti, M.~Roderick, M.~Fazlyab, and J.~Z. Kolter, ``Enforcing robust control guarantees within neural network policies,'' {\em arXiv preprint arXiv:2011.08105}, 2020.

\bibitem{srinivasan2020learning}
K.~Srinivasan, B.~Eysenbach, S.~Ha, J.~Tan, and C.~Finn, ``Learning to be safe: Deep rl with a safety critic,'' {\em arXiv preprint arXiv:2010.14603}, 2020.

\bibitem{alshiekh2018safe}
M.~Alshiekh, R.~Bloem, R.~Ehlers, B.~K{\"o}nighofer, S.~Niekum, and U.~Topcu, ``Safe reinforcement learning via shielding,'' in {\em Proceedings of the AAAI Conference on Artificial Intelligence}, vol.~32, 2018.

\bibitem{li2020robust}
S.~Li and O.~Bastani, ``Robust model predictive shielding for safe reinforcement learning with stochastic dynamics,'' in {\em 2020 IEEE International Conference on Robotics and Automation (ICRA)}, pp.~7166--7172, IEEE, 2020.

\bibitem{emam2022safe}
Y.~Emam, G.~Notomista, P.~Glotfelter, Z.~Kira, and M.~Egerstedt, ``Safe reinforcement learning using robust control barrier functions,'' {\em IEEE Robotics and Automation Letters}, 2022.

\bibitem{oksendal2013stochastic}
B.~Oksendal, {\em Stochastic differential equations: an introduction with applications}.
\newblock Springer Science \& Business Media, 2013.

\bibitem{oksendal_stochastic_2003a}
B.~{\O}ksendal, {\em Stochastic {{Differential Equations}}: {{An Introduction}} with {{Applications}}}.
\newblock Universitext, {Berlin Heidelberg}: {Springer-Verlag}, sixth~ed., 2003.

\bibitem{borodin_stochastic_2017}
A.~N. Borodin, {\em Stochastic processes}.
\newblock Springer, 2017.

\bibitem{rawlings2017model}
J.~B. Rawlings, D.~Q. Mayne, M.~Diehl, {\em et~al.}, {\em Model predictive control: theory, computation, and design}, vol.~2.
\newblock Nob Hill Publishing Madison, WI, 2017.

\bibitem{Jensen1906}
J.~L. W.~V. Jensen, ``{Sur les fonctions convexes et les in{\'{e}}galit{\'{e}}s entre les valeurs moyennes},'' {\em Acta Mathematica}, vol.~30, pp.~175--193, 1906.

\bibitem{zhang2023optimal}
Q.~Zhang, A.~Taghvaei, and Y.~Chen, ``An optimal control approach to particle filtering,'' {\em Automatica}, vol.~151, p.~110894, 2023.

\bibitem{thijssen2015path}
S.~Thijssen and H.~Kappen, ``Path integral control and state-dependent feedback,'' {\em Physical Review E}, vol.~91, no.~3, p.~032104, 2015.

\bibitem{wang2023generalizable}
Z.~Wang and Y.~Nakahira, ``A generalizable physics-informed learning framework for risk probability estimation,'' in {\em Learning for Dynamics and Control Conference}, pp.~358--370, PMLR, 2023.

\bibitem{wang2024physics}
Z.~Wang, R.~Keller, X.~Deng, K.~Hoshino, T.~Tanaka, and Y.~Nakahira, ``Physics-informed representation and learning: Control and risk quantification,'' in {\em Proceedings of the AAAI Conference on Artificial Intelligence}, vol.~38, pp.~21699--21707, 2024.

\bibitem{ACC24}
H.~Hoshino and Y.~Nakahira, ``A physics-informed reinforcement learning framework for risk probability estimation,'' in {\em American Control Conference}, 2024.

\bibitem{wang2025generalizable}
Z.~Wang, A.~Chern, and Y.~Nakahira, ``Generalizable physics-informed learning for stochastic safety-critical systems,'' {\em IEEE Transactions on Automatic Control}, 2025.

\bibitem{park1991universal}
J.~Park and I.~W. Sandberg, ``Universal approximation using radial-basis-function networks,'' {\em Neural computation}, vol.~3, no.~2, pp.~246--257, 1991.

\bibitem{hammersley2013monte}
J.~Hammersley, {\em Monte carlo methods}.
\newblock Springer Science \& Business Media, 2013.

\bibitem{chern2021safe}
A.~Chern, X.~Wang, A.~Iyer, and Y.~Nakahira, ``Safe control in the presence of stochastic uncertainties,'' in {\em 2021 60th IEEE Conference on Decision and Control (CDC)}, pp.~6640--6645, IEEE, 2021.

\bibitem{prandini2006stochastic}
M.~Prandini and J.~Hu, ``A stochastic approximation method for reachability computations,'' in {\em Stochastic hybrid systems: theory and safety critical applications}, pp.~107--139, Springer, 2006.

\bibitem{crank1947practical}
J.~Crank and P.~Nicolson, ``A practical method for numerical evaluation of solutions of partial differential equations of the heat-conduction type,'' in {\em Mathematical proceedings of the Cambridge philosophical society}, vol.~43, pp.~50--67, Cambridge University Press, 1947.

\bibitem{marcus2023constrained}
E.~Marcus, R.~Sheombarsing, J.-J. Sonke, and J.~Teuwen, ``Constrained empirical risk minimization: Theory and practice,'' {\em arXiv preprint arXiv:2302.04729}, 2023.

\bibitem{parrilo2000structured}
P.~A. Parrilo, {\em Structured semidefinite programs and semialgebraic geometry methods in robustness and optimization}.
\newblock California Institute of Technology, 2000.

\bibitem{lasserre2009moments}
J.~B. Lasserre, {\em Moments, positive polynomials and their applications}, vol.~1.
\newblock World Scientific, 2009.

\bibitem{garcia1989model}
C.~E. Garcia, D.~M. Prett, and M.~Morari, ``Model predictive control: Theory and practice—a survey,'' {\em Automatica}, vol.~25, no.~3, pp.~335--348, 1989.

\bibitem{mayne2000constrained}
D.~Q. Mayne, J.~B. Rawlings, C.~V. Rao, and P.~O. Scokaert, ``Constrained model predictive control: Stability and optimality,'' {\em Automatica}, vol.~36, no.~6, pp.~789--814, 2000.

\bibitem{sutton1999policy}
R.~S. Sutton, D.~McAllester, S.~Singh, and Y.~Mansour, ``Policy gradient methods for reinforcement learning with function approximation,'' {\em Advances in neural information processing systems}, vol.~12, 1999.

\bibitem{watkins1992q}
C.~J. Watkins and P.~Dayan, ``Q-learning,'' {\em Machine learning}, vol.~8, pp.~279--292, 1992.

\bibitem{sutton2018reinforcement}
R.~S. Sutton and A.~G. Barto, {\em Reinforcement learning: An introduction}.
\newblock MIT press, 2018.

\bibitem{gaskett1999q}
C.~Gaskett, D.~Wettergreen, and A.~Zelinsky, ``Q-learning in continuous state and action spaces,'' in {\em Australasian joint conference on artificial intelligence}, pp.~417--428, Springer, 1999.

\bibitem{ahmadi2020risk}
M.~Ahmadi, X.~Xiong, and A.~D. Ames, ``Risk-sensitive path planning via cvar barrier functions: Application to bipedal locomotion,'' {\em arXiv preprint arXiv:2011.01578}, 2020.

\bibitem{pereira2021safe}
M.~Pereira, Z.~Wang, I.~Exarchos, and E.~Theodorou, ``Safe optimal control using stochastic barrier functions and deep forward-backward sdes,'' in {\em Conference on Robot Learning}, pp.~1783--1801, PMLR, 2021.

\bibitem{batra2022decentralized}
S.~Batra, Z.~Huang, A.~Petrenko, T.~Kumar, A.~Molchanov, and G.~S. Sukhatme, ``Decentralized control of quadrotor swarms with end-to-end deep reinforcement learning,'' in {\em Conference on Robot Learning}, pp.~576--586, PMLR, 2022.

\bibitem{dixit2021risk}
A.~Dixit, M.~Ahmadi, and J.~W. Burdick, ``Risk-sensitive motion planning using entropic value-at-risk,'' in {\em 2021 European Control Conference (ECC)}, pp.~1726--1732, IEEE, 2021.

\bibitem{schneider2024learning}
L.~Schneider, J.~Frey, T.~Miki, and M.~Hutter, ``Learning risk-aware quadrupedal locomotion using distributional reinforcement learning,'' in {\em 2024 IEEE International Conference on Robotics and Automation (ICRA)}, pp.~11451--11458, IEEE, 2024.

\bibitem{jing2022probabilistic}
H.~Jing and Y.~Nakahira, ``Probabilistic safety certificate for multi-agent systems,'' in {\em 2022 IEEE 61st Conference on Decision and Control (CDC)}, pp.~5343--5350, IEEE, 2022.

\bibitem{gangadhar2023occlusion}
S.~Gangadhar, Z.~Wang, K.~Poku, N.~Yamada, K.~Honda, Y.~Nakahira, H.~Okuda, and T.~Suzuki, ``An occlusion- and interaction-aware safe control strategy for autonomous vehicles,'' in {\em 2023 22nd IFAC World Congress}, 2023.

\end{thebibliography}
